\RequirePackage{fix-cm}
\documentclass[smallextended]{svjour3}       
\smartqed 
\usepackage{graphicx}
\usepackage[shortcuts]{extdash}

\usepackage{amsmath,amsfonts,amssymb,amstext,bbm}

\usepackage{natbib}
\usepackage{xcolor}

\usepackage{float}
\floatstyle{ruled}
\newfloat{algorithm}{tbp}{loa}
\floatname{algorithm}{Algorithm}

\usepackage{multirow}
\usepackage{balance} 

\newcommand{\allitems}{\mathcal{M}}
\newcommand{\allagents}{\mathcal{N}}
\newcommand{\allincomes}{\mathbb{R}_{+}^n}

\newcommand{\allincomesplus}{\mathbb{R}_{+}^{n+1}}
\newcommand{\allpreferences}{\mathbb{M}^{n,m}}
\newcommand{\allpreferencesplusm}{\mathbb{M}^{n,m+1}}
\newcommand{\allpreferencesplusn}{\mathbb{M}^{n+1,m}}

\newcommand{\incomevector}{\mathbf{t}}
\newcommand{\pricevector}{\mathbf{p}}

\newcommand{\range}[2]{\in\{#1,\dots,#2\}}

\newcommand{\pick}[2]{
\hskip 3mm
^{
	{\displaystyle \textbf{\color{blue} #1}}
}_{
	{\displaystyle {\color{red!70} #2 \phantom{a^+}}}
}
}

\newcommand{\pixep}[2]{
\text{If } & #1 && \text{then } && #2 ~~~~.
}

\newcommand{\partition}[2]{\textsc{Partition}(#1,#2)}
\newcommand{\union}[2]{\textsc{Union}(#1,#2)}
\newcommand{\maxmin}[3]{\left[\frac{#2}{#3}\right]#1}

\journalname{}

\begin{document}

\title{Competitive Equilibrium For Almost All Incomes: Existence and Fairness
\thanks{A preliminary version appeared in the proceedings of AAMAS 2018. 
That version contained a severe error --- the algorithm for 4 goods and 3 agents had a bug. 
I am very grateful to an anonymous referee for a comment that lead to revealing this bug (as well as hundreds of other helpful comments).
In fact, there is an impossibility result for this case, which is presented in this paper.
Moreover, this paper extends the results to allocation of indivisible chores (bads). Additionally, it proves new fairness criteria satisfied by competitive equilibria.}
}

\author{Erel Segal-Halevi}

\institute{Erel Segal-Halevi \at
           Ariel University \\
           Ariel 40700 \\
           Israel \\
           \email{erelsgl@gmail.com}
}

\date{Received: date / Accepted: date}

\maketitle

\begin{abstract}
Competitive equilibrium (CE)
is a fundamental concept in market economics. 
Its efficiency and fairness properties 
make it particularly appealing as a rule for fair allocation of resources among agents with possibly different entitlements.
However, when the resources are indivisible, a CE might not exist even when there is one resource and two agents with equal incomes.
Recently, Babaioff and Nisan and Talgam-Cohen (2017) have suggested to consider the entire space of possible incomes, and check whether there exists a \emph{competitive equilibrium for almost all income-vectors} --- all income-space except a subset of measure zero. 
They proved various existence and non-existence results, but left open the cases of four goods and three or four agents with monotonically-increasing preferences.

This paper proves non-existence in both these cases, thus completing the characterization of CE existence for almost all incomes in the domain of monotonically increasing preferences. 
Additionally, the paper provides a complete characterization of CE existence in the domain of monotonically decreasing preferences, corresponding to allocation of chores.

On the positive side, the paper proves that CE exists for almost all incomes when there are four goods and three agents with additive preferences.
The proof uses  a new tool for describing a CE, as a subgame-perfect equilibrium of a specific sequential game.  The same tool also enables substantially simpler proofs to the cases already proved by Babaioff et al.

Additionally, this paper proves several strong fairness properties that are satisfied by any CE allocation, illustrating its usefulness for fair allocation among agents with different entitlements.

\keywords{Competitive equilibrium \and Fair division \and Subgame perfect equilibrium \and Picking sequence \and Maximin share \and Indivisible goods \and Indivisible chores}
\end{abstract}


\section{Introduction}
Competitive equilibrium (henceforth CE)%
\footnote{also called \emph{Walrasian equilibrium} or \emph{price equilibrium} or \emph{market equilibrium}.} 
is a fundamental concept in economics.
It is interesting from both a descriptive and a normative viewpoint.

From a descriptive viewpoint, CE is a condition of stability in a market \citep{walras1874elements}. 
In the simple model called \emph{Fisher market} \citep{brainard2000compute,branzei2014fisher}, 
a single seller comes to the market with some items for sale. 
Several buyers with different preferences over bundles come to the market. Each buyer $i$ carries a certain amount $t_i$ of fiat money called \emph{income}%
\footnote{also called \emph{budget}.}.
CE is defined as a price-vector and a partition of the items among the buyers, such that the bundle of each buyer $i$ costs at most $t_i$ and is weakly preferred by $i$ to any other bundle that costs at most $t_i$. The existence of a CE implies that the market can attain a ``steady state'' in which the total demand (i.e., the union of the preferred bundles of all the buyers) equals the supply.

From a normative viewpoint, CE can be considered a rule for fair allocation.
In the basic setting, some items have to be allocated among several agents with equal entitlements. The goal is to find an allocation that is both \emph{weakly Pareto-efficient} (no other allocation is strictly preferred by all agents) and \emph{envy-free} (each agent weakly prefers his/her bundle to the bundle of every other agent). It is known that these two properties are satisfied by any CE allocation
with \emph{equal incomes} (CEEI)
\citep{Varian1974Equity,Moulin2004Fair}.
In general, agents may have different entitlements, where the entitlement of agent $i$ is represented by a positive number $t_i$ (for example, the items may belong to a firm of which each agent $i$ owns $t_i$ shares).
Then, the goal is to find an efficient allocation that satisfies some appropriate generalization of envy-freeness. 
For example, 
\citet{Reijnierse1998Finding} define an allocation to be 
\emph{$\mathbf{t}$-envy-free} (where $\mathbf{t}$ is the vector of entitlements) if each agent $i$ values his/her share as at least $t_i/t_j$ of the share of agent $j$.%
\footnote{
The agents in their model have cardinal valuation functions.
}
They show that, in any CE in which the income of each agent $i$ equals his/her entitlement $t_i$, the allocation is weakly Pareto-efficient and $\mathbf{t}$-envy-free.

When all items are divisible, 
whether the incomes are equal or different, a
CE exists under very mild assumptions, both for homogeneous divisible resources \citep{wald1935uber,Arrow1954Existence,mckenzie1954equilibrium} and for a heterogeneous resource such as a cake or a land-estate \citep{Weller1985Fair,Reijnierse1998Finding,segalhalevi2018monotonicity}.
But when the items are \emph{indivisible}, a CE
might fail to exist even in very simple cases. For example, when there is one good and two agents with equal incomes, no CE exists: if the price is larger than the income, no agent can afford the good; otherwise, both agents can afford it, but only one agent can get it
(Note that in this model money has no intrinsic value, so an agent always \emph{strictly} prefers to buy an affordable good than to remain with no goods).

The example above could make us think that we cannot enjoy the benefits of CE when there are indivisible items. But a recent paper by \citet{Babaioff2017Competitive} gives a new hope. 
They start by noticing that, in the case of two agents and one item, a CE fails to exist \emph{only when the incomes are exactly equal}. So one can say that \emph{a CE exists for almost all incomes} --- the set of income-vectors in which a CE does not exist has a (Lebesgue) measure of zero.

They then ask whether this existence-for-almost-all-incomes extends to situations with more than one item and more than two agents.
As with divisible items, the question is interesting from both a positive and normative viewpoint: whenever a CE exists, the CE allocation is weakly Pareto-efficient and has appealing fairness properties that can be seen as a discrete analogue of $\incomevector$-envy-freeness (see Section \ref{sec:fairness} below).
Hence, the existence of CE for almost all incomes would imply that almost all instances of item allocation (with different entitlements) have a fair and efficient solution.
The hope of solving such allocation problems motivates the following research question:
\begin{center}
\emph{In what cases does a CE exist for almost all incomes?}
\end{center}
Here, ``case'' is represented by the number of items, the number of agents and the class of the agents' preferences.
\citet{babaioff2019fair} considered the case of an arbitrary number of items and two agents with \emph{identical additive} preferences, where each item has a value and both agents prefer the bundle in which the sum of values is higher. They proved that in this case a CE exists for almost all incomes.%
\footnote{
They also proved that, when the two agents have additive but possibly different preferences, a CE exists for ``almost equal'' incomes (almost all incomes in a neighborhood of the vector of equal incomes). The general case, in which both the preferences and the incomes are different, remains open.
}

In parallel \citep{babaioff2019competitive}, they  considered the more general class of \emph{monotonically-increasing} preferences, when the only restriction is that each agent prefers a set to all its subsets (in this case the items are called ``goods''). 
They proved that, when there are at most 3 goods, or at most 4 goods and 2 agents, for any monotonically-increasing preference-profile, 
a CE exists for almost all income-vectors.
On the other hand, they proved a negative result for the case of 5 goods and 2 agents: they presented a profile of monotonically-increasing preferences for which the subset of the income-space where a CE does not exist has a strictly positive measure. 
They left two cases open: 4 goods and 3 agents, and 4 goods and 4 agents.

\subsection{Contributions}
The first contribution of this paper is a negative answer to the two open cases. 
For 4 goods and 3 or more agents, 
there exists a profile 
of monotonically-increasing preferences
and a positive-measure subset of the income-space, in which a CE does not exist.%
\footnote{
In the conference version \citep{SegalHalevi2018CEFAI}
I mistakenly claimed that the answer is ``Yes'' with 4 goods and 3 agents,
and becomes ``No'' only with 4 goods and 4 or more agents.
}
Thus, the existence of a CE for almost-all incomes for monotonically-increasing preferences is summarized by the following table (where stars denote new results):

\newcommand{\resultref}[1]{{\small (Sec.\ref{#1})}}

\begin{center}
\begin{tabular}{|c|c|c|c|}
\hline 
\# Goods:       & $1,2,3$  & $4$ & $5+$ \\ 
\hline 
$2$ agents:  & \multirow{3}{20mm}{Yes \resultref{sec:3items}} & Yes \resultref{sub:4items-2agents}& \multirow{3}{20mm}{No \resultref{sec:5items}} \\ 
\cline{0-0}
\cline{3-3}
$3$ agents:  &  & \textbf{No* \resultref{sub:4items-3agents}} &  \\ 
\cline{0-0}
\cline{3-3}
$4+$ agents: &  & \textbf{No* \resultref{sub:4items-4agents}} &  \\ 
\hline 
\end{tabular}
\end{center}
The effort to solve the open cases yielded a tool that may be useful in its own right:
For the cases in which a CE exists (the cases marked by ``Yes'' in the above table), 
there is a sequential game, in which each agent in turn picks an item, and at least one subgame-perfect equilibrium of the game is a CE. 

Using this tool,
the paper presents a new algorithm that finds a CE in the case of $4$ goods and $3$ agents with \emph{additive} preferences.
In contrast, the impossibility result for $4$ goods and $4$ agents holds even with additive preferences. The following table summarizes what is currently known about additive preferences:
\begin{center}
\begin{tabular}{|c|c|c|c|}
\hline 
\# Goods:       & $1,2,3$  & $4$ & $5+$ \\ 
\hline 
$2$ additive agents:  & \multirow{3}{20mm}{Yes \resultref{sec:3items}} & Yes \resultref{sub:4items-2agents} & \multirow{2}{20mm}{Open} \\ 
\cline{0-0}
\cline{3-3}
$3$ additive agents:  &  & \textbf{Yes* \resultref{sub:4items-3agents-additive}} & \\ 
\cline{0-0}
\cline{3-4}
$4+$ additive agents: &  & \multicolumn{2}{|c|}{\textbf{No* \resultref{sub:4items-4agents}}}  \\ 
\hline 
\end{tabular}
\end{center}


The paper also considers \emph{monotonically-decreasing} preferences, when each agent prefers a set to all its supersets (the items in this case are called ``bads'' or ``chores'').
With two agents, there is an equivalence between goods-allocation and chores-allocation. It implies that both positive and negative results for allocating goods between two agents are true for chores too. With three or more agents, a CE for almost-all incomes might not exist with any number of chores, even with additive preferences. So the tables of results for chores are:
\begin{center}
\begin{tabular}{|c|c|c|}
\hline 
\# Chores \resultref{sec:chores}:   & $\leq 4$ & $5+$ 
\\ 
\hline 
$2$ agents:  & \textbf{Yes*} & \textbf{No*}
\\ 
\hline 
$3+$ agents:  &  \multicolumn{2}{|c|}{\textbf{No*}}   \\ 
\hline 
\end{tabular} 
~~
\begin{tabular}{|c|c|c|}
\hline 
\# Chores \resultref{sec:chores}:   & $\leq 4$ & $5+$ 
\\ 
\hline 
$2$ additive agents:  & \textbf{Yes*} & {Open}
\\ 
\hline 
$3+$ additive agents:  &  \multicolumn{2}{|c|}{\textbf{No*}}   \\ 
\hline 
\end{tabular} 
\end{center}

An additional contribution of this paper is an analysis of various fairness properties that are satisfied by a CE allocation --- properties that are natural generalizations of envy-freeness for agents with different entitlements \resultref{sec:fairness}.
It is shown that, in general, these new fairness properties are independent --- none of them implies the other one \resultref{app:fairness}.
An allocation satisfying these properties is called \emph{CE-fair}. Every CE allocation is  CE-fair, but the opposite is not true \resultref{sec:ce-fair-not-ce}. 

In general, a CE-fair allocation may exist even when a CE does not exist.
However, all the results on non-existence of CE in the present paper are proved using 
only the CE-fairness properties. 
Therefore, they prove that even an allocation that is only CE-fair is not guaranteed to exist for almost all incomes. 

\subsection{Paper layout}
The formal definitions are presented in Section 
\ref{sec:model}.
The fairness properties of CE are presented in 
Section \ref{sec:fairness}. 
The sequential game used to find a CE for almost all incomes is presented in Section \ref{sec:pixep}.
The settings of three, four and five goods are analyzed in Sections \ref{sec:3items}, \ref{sec:4items} and \ref{sec:5items} respectively. 
Allocation of chores is discussed in 
Section \ref{sec:chores}.
Related work is surveyed in Section \ref{sec:related} and future work ideas are presented in Section \ref{sec:future}.
Some side-results are presented in the appendices.

\section{Preliminaries}
\label{sec:model}
There is a set $\allagents{}$ of agents, with $n=|\allagents|$.
The agents are denoted $i\range{1}{n}$ or Alice, Bob, Carl, etc.
Each agent has a pre-determined positive \emph{income}.
The incomes are denoted $t_i$ (for $i\range{1}{n}$), or $a, b, c,$ etc. 
An income-vector is denoted by $\incomevector$;
the set of all possible income vectors for $n$ agents is $\allincomes$.
There is a set $\allitems$ of items with $m=|\allitems|$. Items are named $z,y,x,$ etc. 
A \emph{bundle} is a set of items. For brevity, a bundle is often written as a string of its items. E.g, $xy$ represents the bundle $\{x,y\}$.

\subsection{Preferences}
\label{sub:preferences}
Each agent $i$ has a complete and transitive  preference-relation $\succeq_i$ over bundles. 
The relation $\succ_i$ is defined as usual by: $X\succ_i Y$ iff $X\succeq_i Y$ and not $Y\succeq_i X$. 
An $n$-tuple of preference-relations is called a \emph{preference-profile}.
A preference-relation $\succeq_i$ is called:
\begin{itemize}
\item \emph{Strict} --- if for every two bundles $X\neq Y$, either $X\succ_i Y$ or $Y\succ_i X$ (no agent is indifferent between two bundles).
\item \emph{Monotonically increasing} --- if
$X \supseteq Y$ implies $X \succeq_i Y$ (an agent prefers a bundle over all its subsets).
\item \emph{Monotonically decreasing} --- if
$X \subseteq Y$ implies $X \succeq_i Y$ (an agent prefers a bundle over all its supersets).
\item \emph{Additive} --- if there exists 
a measure (an additive set function) $V_i: 2^{\allitems}\to\mathbb{R}$ such that,
for every two bundles $X,Y$:
$X\succeq_i Y$ iff 
$V_i(X) \geq V_i(Y)$.
\end{itemize}
In most of the paper, it is assumed that the preferences and {monotonically\-/increasing}, while 
in Section \ref{sec:chores}, it is assumed that the 
preferences are {monotonically\-/decreasing}.
Strictness and additivity are assumed only in  specific results.

Given $n$ agents and $m$ items, the set of all profiles of monotonically-increasing preferences is denoted by $\allpreferences$.%

\subsection{Allocations}
An \emph{allocation} $\mathbf{X}$ is a partition of $\allitems$ among the $n$ agents.
A \emph{price-vector} $\pricevector$ is a vector of positive numbers, one number per item. The price of a bundle $Z\subseteq\allitems$ is denoted $p(Z)$ and it equals the sum of the prices of its items: $p(Z) = \sum_{z\in Z}{p_z}$.

\begin{definition}
\label{def:ce}
Given an income-vector $\incomevector$,
a \emph{competitive equilibrium (CE) w.r.t $\incomevector$} is a pair $(\pricevector,\mathbf{X})$, where $\pricevector$ is a price-vector and $\mathbf{X}$ is an allocation and the following conditions hold for every agent $i\in\allagents$.

\textbf{Condition 1.}
The price of the agent's bundle is at most the agent's income:
\begin{align*}
p(X_i)\leq t_i .
\end{align*}

\textbf{Condition 2.}
The agent's bundle is the best bundle he/she can afford with his/her income.
Formally, for every bundle $Y\subseteq \allitems$:
\begin{align*}
p(Y)\leq t_i ~~~ \implies~~~ Y\preceq_i X_i .
\end{align*}
\end{definition}

\subsection{Existence of CE}
This paper focuses on the following definitions of existence of CE.
\begin{definition}
\label{def:ce-exists-in-t}
Given integers $n\geq 2$ and $m\geq 1$ and an income-vector $\incomevector \in \allincomes$, \emph{a CE exists for $\incomevector$} if for all preference-profiles in $\allpreferences$, there is a price-vector $\pricevector$ and an allocation $\mathbf{X}$ such that $(\pricevector,\mathbf{X})$ is a CE given the income-vector $\incomevector$. 
\end{definition}

Throughout, the term \emph{measure} denotes the Lebesgue measure on $\mathbb{R}^n$.

\begin{definition}
\label{def:ce-exists-in-n-m}
Given integers $n\geq 2$ and $m\geq 1$, 
\emph{a CE exists for almost all incomes} 
if the set of income-vectors for which no CE exists has a zero measure.%
\footnote{
The positive results in this paper 
hold even for a slightly stronger notion of existence, by which
the set of income-vectors for which no CE exists is \emph{nowhere dense} in $\allincomes$.
}
In other words: 
for almost all $\incomevector \in \allincomes$, for all preference-profiles in $\allpreferences$, a CE exists.%
\footnote{
One could consider alternative definitions by changing the order of quantifiers. Such alternative definitions are discussed in Appendix \ref{sec:alternative-defs}
}
\end{definition}
To show existence by Definition \ref{def:ce-exists-in-n-m}, it is sufficient to show a finite set of equalities on the incomes such that a CE exists whenever none of these equalities is satisfied.

Note that Definition \ref{def:ce-exists-in-n-m} allows the CE to depend on the income-vector. I.e, for every income-vector $\incomevector\in\allincomes$ (except a set of measure zero), there may be a different CE allocation and price-vector.

In some proofs of existence, it will be convenient to assume that the preferences are strict. 
The following lemma shows that, regarding the existence of a CE, this assumption does not lose generality.

\begin{lemma}
\label{lem:strict-is-wlog}
Given integers $n\geq 2$ and $m\geq 1$ and an income-vector $\incomevector \in \allincomes$, if a CE exists for $\incomevector$ for all \emph{strict} preference-profiles in $\allpreferences$,
then a CE exists for  
$\incomevector$ for all preference-profiles in $\allpreferences$.
\end{lemma}
\begin{proof}
Consider a preference-profile that is not strict, so that for some agents $i\in\allagents$ and some pairs of bundles $X,Y$, both $X\succeq_i Y$ and $Y\succeq_i X$ hold.
For each such agent, the relation $\succeq_i$ can be made strict by arbitrarily removing either $X\succeq_i Y$ or $Y\succeq_i X$. 
By assumption, there exists a CE $(\mathbf{p},\mathbf{X})$ for the new profile. 

This $(\mathbf{p},\mathbf{X})$  is a CE for the original profile too: Condition 1 holds trivially since it does not depend on the preferences. As for Condition 2, 
if $p(Y)\leq t_i$ in the original profile then the same is true in the new profile, which implies that in the new profile $Y\preceq_i X$. Since the new profile did not add any new relations (only removed some), the relation 
$Y\preceq_i X$
must hold in the original profile too.\qed
\end{proof}

\section{Fairness Properties of Competitive Equilibria}
\label{sec:fairness}
It is well-known \citep{Varian1974Equity,bouveret2016characterizing} that, 
when all incomes are equal, 
a CE allocation is always \emph{envy-free} --- each agent weakly prefers his/her own bundle to the bundle of any other agent
($\forall i,j\in\allagents: X_i \succeq_i X_j$).
\citet{budish2011combinatorial} introduced another fairness property that is always satisfied by CE with equal incomes: the CE allocation guarantees to each agent his/her \emph{1-out-of-$n$ maximin share}, which is the best bundle an agent can get by partitioning the items into $n$ bundles and getting the worst one \citep{budish2011combinatorial}[Proposition 4].
This section explores the fairness properties of CE when the incomes are general --- not necessarily equal.

\citet{babaioff2019fair}[Proposition 3.3] have already generalized the result of \citet{budish2011combinatorial} by proving that, if an agent's income is at least a fraction $l/d$ of the sum of all incomes (where $l,d$ are positive integers), then a CE allocation guarantees him his  \emph{$l$-out-of-$d$ maximin-share}, which is the best bundle he can get by partitioning the items into $d$ bundles and getting the worst $l$ of them (see formal definition below).
This proposition supports the use of CE as a rule for fair allocation among agents with different entitlements:
Given an allocation problem with an entitlement vector $\mathbf{t}$,
construct a Fisher market with an income vector 
$\mathbf{t}$, find a CE $(\mathbf{X},\pricevector)$, and implement the allocation $\mathbf{X}$.

But CE has many more fairness implications, which are not captured by the proposition of \citet{babaioff2019fair}.
For example, as shown below, a CE guarantees that, if Alice's income is at least as large as Bob's income, then Alice does not envy Bob; if it is at least as large as half Bob's income, then Alice does not envy ``half'' of Bob's bundle; and 
if it is at least as large as twice Bob's income, 
then Alice does not envy ``twice'' of Bob's bundle. These fairness properties are formalized and proved below.

\subsection{Notation}
This section focuses on a fixed CE allocation $(\pricevector,\mathbf{X})$ and a specific agent, Alice. Her income is denoted by $a$,  her preference-relation by $\succeq$, and her bundle in the fixed CE by $A$.

For every bundle $Z\subseteq\allitems$ and integer $d\geq 1$, $\partition{Z}{d}$ denotes the set of all partitions of $Z$ into $d$ sub-bundles (some possibly empty).

For every bundle-vector $\mathbf{Y}=(Y_1,\ldots,Y_d)$ and integer $l\range{1}{d}$, $\union{\mathbf{Y}}{l}$ is the set of all unions of exactly $l$ different bundles from $\mathbf{Y}$, $Y_{j1}\cup Y_{j2} \cup \cdots \cup Y_{jl}$.

For every bundle $Z$ and integers $l,d$ with $1\leq l\leq d$, the \emph{$l$-out-of-$d$-maximin-bundle} of $Z$  is denoted $\maxmin{Z}{l}{d}$ and defined as:
\footnote{
The notation ``$l$-over-$d$'' is introduced here in order to make the properties \eqref{eq:1} and \eqref{eq:2} below visually symmetric and easy to remember.
}
\begin{align*}
\maxmin{Z}{l}{d}
~~
:= 
~~
\max_{\mathbf{Y}\in \partition{Z}{d}}
~~
\min_{W\in \union{\mathbf{Y}}{l}}
~~
W
\end{align*}
where $\max,\min$ are based on Alice's preference-relation $\succ$.


In the above notation, Prop.  4 of \citet{budish2011combinatorial} says that:
\begin{align*}
\text{All incomes are equal}
&& \implies &&
A \succeq \maxmin{\allitems}{1}{n}
\end{align*}
Prop. 3.3 of \citet{babaioff2019fair} says that, for every two integers $1\leq l\leq d$:
\begin{align*}
a \geq {l \over d} \left(\sum_{i\in \allagents} t_i\right)
&& \implies &&
A \succeq \maxmin{\allitems}{l}{d}
\end{align*}

\begin{remark}
If Alice's preferences are strict, then  $\maxmin{Z}{l}{d}$ is a unique bundle.
Otherwise, there may be several maximizing bundles.
The results in this section are specified, for simplicity, as though $\maxmin{Z}{l}{d}$ is a unique bundle. However, it is easy to see that they actually hold for \emph{any} bundle satisfying the definition of $\maxmin{Z}{l}{d}$, since Alice is indifferent among all these bundles.
\end{remark}

\subsection{First generalized fairness property}
For every subset of agents $K\subseteq \allagents$,
if we consider only the subset of items $\left(\bigcup_{i\in K} X_i\right)$, 
with the same prices that these items have in $\pricevector$, the allocation still satisfies the two CE conditions.
%
 Hence, Prop. 3.3 of \citet{babaioff2019fair} implies that, for every subset of agents $K\subseteq \allagents$ which contains Alice,
and for every two integers $l,d$ with $1\leq l\leq d$:
\begin{align}
\label{eq:1}
\tag{P1}
	a \geq {l \over d} \left(\sum_{i\in K} t_i\right)
	&& \implies &&
	A \succeq {\maxmin{\left(\bigcup_{i\in K} X_i\right)}{l}{d}}
\end{align}
Note that \eqref{eq:1} reduces to Prop. 3.3 of \citet{babaioff2019fair} when  $K=\allagents$, 
which reduces to Prop.  4 of \citet{budish2011combinatorial} when $l=1,d=n$.
For completeness, a stand-alone proof of \eqref{eq:1} (which holds for any $K\subseteq \allagents$) is presented below.
\begin{proposition}
\label{prop:1}
Every CE allocation $\mathbf{X}$ satisfies \eqref{eq:1} for all $K\subseteq \allagents$ and integers $1\leq l\leq d$.
\end{proposition}
\begin{proof}
	Let
	$P$ be the price of the union at the right-hand side of \eqref{eq:1}, $P := p\left({\bigcup_{i\in K} X_i}\right)$.
	By CE Condition 1, for every $i$, $t_i\geq p(X_i)$. Therefore, 
	$
	\sum_{i\in K} t_i
	\geq
	P
	$.
	By the proposition assumption,
	$a \geq {l \over d} 
	\cdot 
	P
	$.
	Consider a partition $\mathbf{Y}\in\partition{\bigcup_{i\in K} X_i}{d}$.
	Order the $d$ parts in $\mathbf{Y}$ by increasing price, i.e, $p(Y_1)\leq \cdots \leq p(Y_d)$. Then, $p(Y_1) + \cdots + p(Y_l) \leq {l\over d} \cdot P$. 
	Define $W :=Y_1\cup \cdots \cup Y_l$. Then, $p(W) \leq {l\over d} \cdot P \leq a$, so Alice can afford $W$. 
	Hence, CE Condition 2 implies that $A\succeq W$.
	Since $W\in \union{\mathbf{Y}}{l}$, by definition $W\succeq {\maxmin{\left(\bigcup_{i\in K} X_i\right)}{l}{d}}$. By transitivity, 
	$A\succeq {\maxmin{\left(\bigcup_{i\in K} X_i\right)}{l}{d}}$.
	\qed
\end{proof}

Proposition \ref{prop:1}
implies some of the no-envy properties mentioned at the start of this section. 
For example, consider an agent Bob, with bundle $B$ and income $b$.
If $a\geq b$, then applying \eqref{eq:1} with $K=\{Bob\}$ and $l=d=1$ gives that $A\succeq \maxmin{B}{1}{1}$. But $\maxmin{B}{1}{1}=B$, so Alice does not envy Bob. 

Moreover, if $a \geq \frac{1}{2}b$, then \eqref{eq:1} implies that Alice weakly prefers her bundle to $\maxmin{B}{1}{2}$ --- her 1-out-of-2 maximin-share of Bob's bundle.

Moreover, if $a\geq (b+c)$, then applying \eqref{eq:1} with
$K=\{Bob,Carl\}$ and $l=d=1$ gives that Alice weakly prefers her bundle to the union of Bob and Carl's bundles. This is stronger than just saying that Alice does not envy Bob or Carl.

When all incomes are equal, 
a CE allocation guarantees Alice 
her \emph{pairwise-maximin-share (PMMS)}, as recently defined by
\citet{caragiannis2019unreasonable}.
In our notation, the PMMS of Alice w.r.t. another agent Bob is simply $\maxmin{\left( A\cup B\right)}{1}{2}$.
When the incomes are equal, $a = {1\over 2}(a+b)$, so Alice weakly prefers $A$ to her PMMS w.r.t. Bob.
Similarly, a CE allocation guarantees the \emph{groupwise-maximin-share (GMMS)}, as recently defined by \citet{barman2018groupwise}. In our notation, the GMMS of Alice w.r.t. a subset $K$ of agents is defined as 
$\maxmin{(\cup_{i\in K}X_i)}{1}{|K|}$.
When the incomes are equal, $a = {1\over |K|}\sum_{i\in K}t_i$, so Alice weakly prefers A to her GMMS w.r.t. any $K$.


\begin{remark}

For each $K\subseteq \allagents$,
there are infinitely many pairs $(l,d)$ 
that satisfy the condition of \eqref{eq:1}, 
but only finitely many lead to ``interesting'' fairness guarantees (that are not implied by other guarantees). 
See Appendix \ref{sub:indep-items} for details and the accompanying technical report\footnote{{http://arxiv.org/abs/1912.08763}} for more details.

Property \eqref{eq:1}
can be instantiated with different subsets $K$ of agents. 
In particular, it can be instantiated with a subset $K$ that does not contain Alice, and with the subset $K\cup \{Alice\}$.
In general, these two instantiations are independent --- none of them implies the other. See Appendix \ref{sub:indep-p1-bundles} for details.

\end{remark}

\subsection{Second generalized fairness property}
When Alice's income is much larger than Bob's income ($a\gg b$), property \eqref{eq:1} guarantees (by taking $l=d=1$) that Alice prefers her bundle to Bob's bundle ($A \succeq B$). However, intuitively it seems that, if Alice's entitlement is much larger than Bob's entitlement, she should get much more than just ``at least $B$''. For example, if $a > 2 b$, Alice should get, in some sense, ``at least twice $B$''.

CE indeed implies such a guarantee. 
To define it formally, 
for every set of pairwise\-/disjoint bundles $Z_1,\ldots,Z_k$ and pairs of integers $(l_1,d_1),\ldots,(l_k,d_k)$, with $0\leq l_i\leq d_i$ for each $i\range{1}{k}$,  define:
	\begin{align*}
	&\left(\maxmin{Z_1}{l_1}{d_1}\right)\sqcup\cdots\sqcup \left(\maxmin{Z_k}{l_k}{d_k}\right)
	~~
	:= \\
	~~
	&~~~~~\max_{
		\mathbf{Y^1}\in \partition{Z_1}{d_1}
		,\ldots,
		\mathbf{Y^k}\in \partition{Z_k}{d_k}
	} \bigg(
	~~ \\
	&~~~~~~~~~~\min_{
		W_1 \in \union{\mathbf{Y^1}}{l_1}
		,\ldots,
		W_k \in \union{\mathbf{Y_k}}{l_k}
	} \bigg(
	~~ \\
	&~~~~~~~~~~~~~~~W_1 \cup \cdots \cup W_k
	\bigg) \bigg)
	\end{align*}
	where $\max,\min$ are based on Alice's preference-relation $\succ$.

In words, it is the best bundle that Alice can guarantee to herself by partitioning each  $Z_i$ into $d_i$ parts, letting an adversary pick $d_i-l_i$ parts from each such partition, and taking the remaining $l_i$ parts from each partition.
With this new definition it is possible to present a new fairness property.
For every $n$ pairs of integers $(l_1,d_1),\ldots,(l_n,d_n)$, with $0\leq l_i\leq d_i$ for all $i\range{1}{n}$:
\begin{align}
\label{eq:2}
\tag{P2}
a \geq 
\frac{l_1}{d_1}t_1 + \cdots + \frac{l_n}{d_n}t_n 
&& \implies &&
A \succeq 
\left(\maxmin{X_1}{l_1}{d_1}\right)\sqcup\cdots\sqcup \left(\maxmin{X_n}{l_n}{d_n}\right)
\end{align}
Note the structural similarity between the left-hand side and the right-hand side of \eqref{eq:2}: every term in the left-hand side (which is a number) corresponds to a term in the right-hand side (which is a subset of items).

If $a \geq 2b$, then $a/2 \geq b$ so $a \geq a/2 + b$. Hence, \eqref{eq:2} implies: $A\succeq \left(\maxmin{A}{1}{2}\right)\sqcup B$. I.e., Alice prefers her bundle to the best bundle she can get by dividing her current bundle into two parts, letting an adversary take away one part, and take Bob's entire bundle in addition to the other part.

\begin{proposition}
\label{prop:2}
Every CE allocation $\mathbf{X}$ satisfies \eqref{eq:2}
for all pairs of integers
$(l_1,d_1),\ldots,(l_n,d_n)$, with $0\leq l_i\leq d_i$.
\end{proposition}
\begin{proof}
By CE Condition 1, $t_i\geq p(X_i)$ for all $i\in\allagents$. Hence, by the left-hand side of \eqref{eq:2}:
\begin{align*}
a \geq 
\sum_{i\in \allagents} {l_i\over d_i}p(X_i).
\end{align*}
For every $i\range{1}{n}$, for every partition $\mathbf{Y^i}\in\partition{X_i}{d_i}$,
the $l_i$ cheaper parts cost together at most ${l_i\over d_i}p(X_i)$. Therefore, there exists a bundle $W_i\in \union{\mathbf{Y^i}}{l_i}$ whose price is at most ${l_i\over d_i}p(X_i)$. The union $W_1\cup\cdots \cup W_n$ costs at most $a$.
CE Condition 2 then implies that  $A\succeq W_1\cup\cdots \cup W_n$.\qed
\end{proof}

\begin{remark}
The two properties \eqref{eq:1} and \eqref{eq:2} are independent --- none of them implies the other one. See Appendix 
\ref{sub:indep-p1-p2} for details.
\end{remark}

\begin{remark}
If Alice's preferences are additive, then \eqref{eq:2} implies a simpler property: in any partition of $A$ into two halves, Alice prefers the better half over $B$. For example, if Alice values the items $x,y,z$ at $4,6,7$, then the allocation $(xy,z)$ violates \eqref{eq:2}, since even though Alice does not envy Bob, both her items are worse than Bob's item. Thus it cannot be a CE. On the other hand, the allocation $(xz,y)$ does not violate \eqref{eq:2}, even though Alice values her bundle at less than twice her evaluation of Bob's bundle.

In contrast, if Alice's preferences are not additive, 
then the property ``Alice prefers the best half of $A$ to $B$'' is not necessarily satisfied by a CE.
For example, suppose the incomes satisfy $2 b < a < 3 b$, there are four goods $w,x,y,z$ and the agents' preferences contain the following relations (where a comma means that the relation is irrelevant for the example):
\begin{itemize}
\item Alice: $xyz \succ \text{all other triplets} \succ wx,wy,wz \succ w \succ xy,xz,yz \succ x \succ y \succ z\succ \emptyset$;
\item Bob: $wx,wy,wz \succ w \succ xy,xz,yz \succ x \succ y \succ z\succ \emptyset$.
\end{itemize}
The allocation $(A,B) = (xyz,w)$ with prices $(p_x,p_y,p_z,p_w) = (a/3,a/3,a/3,b)$ is a CE.
But in any partition of $A$ into two subsets, Alice prefers Bob's bundle to \emph{both} subsets. 

\end{remark}

\begin{remark}
It is possible to generalize \eqref{eq:2} even further. Replace each income $t_i$ in the left-hand side with a sum of some subset of the incomes (e.g. $\sum_{j\in K_i}t_j$ for some subset $K_i\subseteq \allagents$), and replace the corresponding bundle $X_i$ in the right-hand side with a union of the corresponding subset of the bundles ($\cup_{j\in K_i}X_j$). The resulting property is cumbersome to write in its full generality, but it is easy to apply in particular cases using the structural similarity between the two sides of \eqref{eq:2}. For example, the following property holds in CE:
\begin{align*}
a \geq {1\over 2}(t_1+t_2) + {1\over 3}(t_3+t_4)
&&\implies&&
A \succeq \left(\maxmin{X_1\cup X_2}{1}{2}\right) \sqcup \left(\maxmin{X_3\cup X_4}{1}{3}\right)
\end{align*}
\end{remark}

To conclude: competitive equilibrium can be seen as a rule for fair allocation. It guarantees, to each agent, a multitude of fairness properties that generalize, and strictly extend, the properties of both envy-freeness and maximin-share-guarantee.

The propositions in this section motivate the following definition.
\begin{definition}
\label{def:ce-fair}
Given an income-vector $\incomevector$,
an allocation $\mathbf{X}$ is called \emph{CE-fair w.r.t. $\incomevector$}
if for every agent, $\mathbf{X}$ satisfies  \eqref{eq:1} for every subset $K\subseteq \allagents$ and integers $1\leq l\leq d$, and satisfies \eqref{eq:2} for all pairs of integers $0\leq l_i \leq d_i$.
\end{definition}
Propositions \ref{prop:1} and \ref{prop:2} imply that any CE allocation is CE-fair.
However, the opposite is not true: there exist allocations that are CE-fair and even Pareto-efficient, but there are no price-vectors with which they form a CE.
See Appendix \ref{sec:ce-fair-not-ce} for details.

\section{Picking-Sequences with Prices}
\label{sec:pixep}
This section presents a tool for constructing algorithms for finding competitive equilibria. 
Recall that, by Lemma \ref{lem:strict-is-wlog}, 
for proving existence of CE for all preference-profiles, it is sufficient to prove existence of  CE for all strict preference-profiles. Therefore, in this and the following sections it is assumed that all preference relations are strict, so that the terms ``agent $i$'s most preferred item'' or ``agent $j$'s worst pair'' are uniquely defined.
\begin{definition}
A \emph{picking-sequence} is a sequence of $m$ agent-names. It is interpreted as a sequential game in which, at each step, the current agent in the sequence may pick a single item.
\end{definition}
For example, with $m=3$  items, a possible picking-sequence is $A B A$, which denotes a game in which Alice picks an item, then Bob picks an item, then Alice receives the last remaining item.

These games are analyzed below assuming \emph{complete information}, i.e, 
the game, the preferences and the rationality of the agents are common knowledge.
The following \emph{backward induction} analysis
\citep{zermelo1913anwendung,msz2013gametheory} is used. The $m$-th picker just picks the single remaining item.
The $m-1$-th picker picks one of the two remaining items that results in a bundle he prefers (given the items he already took). 
For every possible pair of remaining items, it is known what the $m-1$-th picker is going to pick; based on this knowledge, the $m-2$-th picker picks one of the three remaining items that results in a preferred final bundle for her. 
The analysis proceeds in this way back to step $1$. 
Every sequence of picks that results from this process is called a \emph{subgame-perfect equilibrium} (SPE).
\footnote{
It is known that every SPE is also a \emph{Nash equilibrium}, and moreover, a Nash equilibrium is played in each sub-game (including unreachable ones). 
If rationality is common knowledge (i.e., all players are rational, all players know that all players are rational, etc.), then the outcome of a sequential game with perfect information is the one found by backward induction \citep{aumann1995backward}.
}

For example, consider again the game $A B A$. In step 3 Alice takes the last remaining item. In step 2 Bob chooses the single item he prefers. Suppose w.l.o.g. that for Bob: $x\succ y\succ z$, then Bob will never take $z$. 
Alice knows this, and realizes that her bundle will be either $xz$ or $yz$. So in step 1, Alice decides which of these two bundles she prefers, and picks accordingly. For example, if for Alice: $yz \succ xz$, then in the 1st step Alice picks $y$. Then, Bob picks $x$ and Alice gets $z$, and the final allocation is $yz,x$.
In this case there is a second SPE, in which Alice picks $z$ in the 1st step and gets $y$ in the 3rd step; the allocation in both SPE is the same.

\begin{definition}
(a)
A \emph{picking-sequence-with-prices} (\emph{pixep} for short) is a picking-sequence in which a price is attached to each position. The interpretation is that, whenever an agent picks an item, the corresponding price is attached to that item. 

(b) Let $S$ be a pixep
and $Q$ a subgame-perfect equilibrium in the sequential game defined by $S$.
The pair $(S,Q)$ is called an \emph{execution} of the pixep $S$.
The allocation induced by this execution is denoted by $\mathbf{X}(S,Q)$, and the induced price-vector by $\pricevector(S,Q)$.
\end{definition}
For example, with three items, a possible pixep is:
\begin{align}
\tag{*}
\pick{A}{4}~~~\pick{B}{2}~~~\pick{A}{1}
\end{align}
which means that the first item picked by Alice is priced at 4, the item picked by Bob is priced at 2, and the last item received by Alice is priced at 1. 
A pixep can be seen as a shorthand for an allocation rule; (*) is a shorthand for the rule: ``Give Alice her most preferred pair from the two pairs that contain Bob's worst item; price Alice's two items at 1 and 4; give Bob the remaining item and price it at 2''.

\subsection{Using pixeps to find a competitive equilibrium}
\label{sub:pixeps}
The main goal in designing a pixep is to find a CE.
\begin{definition}
\label{def:pixep-finds-ce}
Let $S$ be a pixep of length $m$ and $\incomevector\in\allincomes$ an income-vector.
$S$ \emph{implements a CE given income-vector $\incomevector$} if whenever the income-vector of the agents is $\incomevector$, for \emph{every} preference-profile in $\allpreferences$, there \emph{exists} a SPE $Q$ of the sequential game defined by $S$, such that the allocation $\mathbf{X}(S,Q)$ with the price-vector $\pricevector(S,Q)$ are a CE.
\end{definition}
Our goal now is to develop pixeps that implement a CE in various settings.
Although there are finitely many sequences, there are infinitely many prices, so exhaustively searching the space of all pixeps is infeasible.
Several heuristics for trimming the search-space are presented below. The heuristics are exemplified on the pixep below, where there are three agents: Alice Bob and Carl, with incomes $a, b, c$.
\begin{align}
\tag{**}
\pick{A}{p_1}~~~\pick{B}{p_2}~~~\pick{B}{p_3}~~~\pick{A}{p_4}
\end{align}

The first heuristic is related to CE Condition 1.
\begin{quote}
{(H1)}
\emph{For each agent $i$ that appears at least once in the sequence, 
the sum of prices that appear below $i$ equals $t_i$.}
\end{quote}
For example, pixep (**) satisfies (H1) iff $p_1+p_4 = a$ and $p_2+p_3 = b$.

(H1) guarantees that, in any outcome $(\pricevector,\mathbf{X})$ of the pixep, $p(X_i)=t_i$ for all $i\in\allagents$, so CE Condition 1 is satisfied. In fact, Condition 1 is weaker and allows $p(X_i)\leq t_i$.
However, as observed by \citet{babaioff2019fair}[Claim 1], 
every CE can be converted into an \emph{income-exhausting CE} (in which all agents with non-empty bundles exhaust their incomes) by increasing the price of an arbitrary item in each agent's bundle until the bundle's price equals the agent's income. So (H1) is without loss of generality --- it does not miss useful pixeps.

The next two heuristics are related to CE Condition 2.

\begin{quote}
(H2)
\emph{The sequence of prices should be decreasing, and strictly-decreasing whenever the picking-sequence switches between agents.}
\end{quote}
\begin{quote}
(H3) 
\emph{The last price must be strictly larger than the income of any agent who does not appear in the sequence.}
\end{quote}
For example, pixep (**) satisfies (H2) iff $p_1 > p_2\geq p_3 > p_4$ and (H3) iff $p_4 > c$.

The (H2) heuristic ensures that no agent can afford to switch the item he picked in his turn with a preferred item picked by another agent in a previous turn. 
This is a necessary condition: suppose (H2) is violated, and suppose all agents have the same additive preferences. Then there is a unique SPE, in which each agent in turn picks the remaining item with the highest value. If (H2) is violated, then this SPE violates CE Condition 2.%
\footnote{
For some particular preference-profiles, there may be competitive equilibria that cannot be found by pixeps with decreasing prices, but can be found by pixeps that do not have decreasing prices (see sub. \ref{sub:auto-pixep}). 
However, a pixep that should implement a CE for \emph{all} preference-profiles must have decreasing prices.
}

The (H3) heuristic ensures that an agent who is allocated an empty bundle cannot afford a non-empty bundle. It is necessary for satisfying CE Condition 2. Hence, both (H2) and (H3), too, are without loss of generality.

Henceforth, only pixeps satisfying (H1,H2,H3) are considered.

\subsection{Domination of bundles}
Given an execution $(S,Q)$, a \emph{domination} relation is defined on bundles based on the positions of items in the sequence. Given two different bundles $X\neq Y$,  say that $X$ is \emph{dominated} by $Y$ if there exists an injection $f: X\to Y$ such that, for each item $x\in X$, $f(x)$ appears (weakly) earlier than $x$ in the execution $(S,Q)$. For example, in a sequence of four items, the pair of items in positions \#1 and \#4:
\begin{itemize}
\item Is dominated by the pair of items \#1 and \#2, as well as by the triplet of items \#1 \#3 \#4;
\item Dominates the pair of items \#3 and \#4, 
as well as the singleton containing item \#1;
\item Is unrelated to the triplet of items \#2 \#3 and \#4 (none of them dominates the other).
\end{itemize}
Given an execution $(S,Q)$ and an agent $i$, the \emph{dominating bundles / dominated bundles / unrelated bundles} of $i$ are the bundles that dominate / are dominated by / are unrelated to $X_i$, respectively. CE Condition 2 can be verified for these three types of bundles separately.

\begin{lemma}\label{lem:dominating}
Suppose a pixep $S$ satisfies (H1,H2,H3). Then in any execution $(S,Q)$, no agent can afford a \emph{dominating bundle}.
\end{lemma}
\begin{proof}
Consider an agent $i\in\allagents$.
If $X_i$ is empty, then (H3) implies that any non-empty bundle costs more than $t_i$.

Otherwise, in any bundle dominating $X_i$, each item appears either at the same location or earlier than a corresponding item in $X_i$.
Moreover, by definition a dominating bundle is different than $X_i$ so it has at least one item selected by a different agent than $i$.
Therefore, (H2) implies that it is costs more than $p(X_i)$. (H1) implies that $p(X_i)=t_i$. Hence, $i$ cannot afford a more expensive bundle. \qed
\end{proof}

\begin{lemma}
\label{lem:dominated-max}
Suppose in a pixep $S$ agent $i$ has $k$ turns.
If in an execution $(S,Q)$, agent $i$ does not want any dominated bundle of size $k$,
then in $(S,Q)$ agent $i$ does not want any dominated bundle.
\end{lemma}
\begin{proof}
Consider a bundle $Y$ dominated by $X_i$.
By definition, there is an injection $f: Y\to X_i$ such that, for each item $y\in Y$, $f(y)$ appears (weakly) earlier than $y$ in the execution $(S,Q)$. 
This implies that $|Y|\leq |X_i| = k$.
If $|Y|=k$ then by assumption $i$ prefers $X_i$ to $Y$.
Otherwise, $|Y|<k$ so there are some $k-|Y|$ elements of $X_i$ that do not equal $f(y)$ for any $y\in Y$. 
Adding these elements to $Y$ gives a new bundle $Y'$ of size $k$, that contains $Y$ and is dominated by $X_i$.
By assumption, $i$ prefers $X_i$ to $Y'$, so by monotonicity, $i$ prefers $X_i$ to $Y$. \qed
\end{proof}

\begin{lemma}
\label{lem:dominated}
Suppose in a pixep $S$ all the turns of agent $i$ are in a single contiguous sequence. Then in any SPE $(S,Q)$, agent $i$ does not want any \emph{dominated bundle}.
\end{lemma}
\begin{proof}
Suppose the turns of  $i$ are a contiguous sequence of length $k$. Then, the unique best strategy of $i$ is to pick the best $k$-tuple from among the items remaining on the table.
This tuple is preferred by agent $i$ to any dominated $k$-tuple. By Lemma \ref{lem:dominated-max}, it is preferred by $i$ to any dominated bundle.
\qed
\end{proof}

\begin{example}
\label{exm:dominated}
In both pixeps (*) and (**), Lemma \ref{lem:dominated} is applicable to Bob. In (*), he picks the best remaining item and obviously does not want the other item; in (**), he picks the best remaining pair and does not want any other remaining pair or singleton.
\end{example}

Lemmata \ref{lem:dominating},
\ref{lem:dominated-max}, \ref{lem:dominated} imply that, to verify that a given pixep implements a CE (by Definition \ref{def:pixep-finds-ce}), it is sufficient to check the unrelated bundles of each agent, and the dominated bundles of agents with non-contiguous turns.
Moreover, it is sufficient to check dominated bundles with the same size as the agent's bundle.
Moreover, (H3) implies that it is not needed to check any bundle for an agent who does not appear in the pixep.

\section{Warm-up: Three Goods}
\label{sec:3items}
As a warm-up, this section shows how to design pixeps implementing CE for the case of three goods and any number of agents. \citet{babaioff2019competitive}~[Proposition 3.1] already proved that in this case a CE exists for almost all incomes, but the algorithm presented here (Algorithm \ref{alg:3 items}) is more concise and contains less cases (see Appendix \ref{sec:proof-comparison} for a detailed comparison).

It can be assumed that all incomes are different, since this assumption removes from the income-space a subset of measure zero. So assume w.l.o.g. that $a>b>c>$ all other incomes.

The following paragraphs examine some picking-sequences and check if they can be made into a pixep that implements a CE.
Consider first the sequence $A A A$, giving Alice all three items. (H3) implies that the last price must be $b+\epsilon$ for some $\epsilon>0$.
(H2) implies that the second price must be at least $b+\epsilon$, so we set it to $b+\epsilon$. (H1) implies that the sum of all prices must equal $a$, so we set the first price to $a-2b-2\epsilon$. (H2) implies that $a-2b-2\epsilon \geq b+\epsilon$, which implies that $a \geq 3b+3\epsilon$.

For brevity, from now on the $\epsilon$ will be omitted from the notation: instead of $b+\epsilon$, I will write $b^+$,
instead of $a-2b-2\epsilon$ I will write $(a-2b)^{--}$, etc.
 So the above discussion can be summarized as:
\begin{align*}
\pixep{a > 3 b}
{\pick{A}{(a-2b)^{--}}~~~\pick{A}{b^+}~~~\pick{A}{b^+}}
\end{align*}
The interpretation of this notation is: 
``If $a>3b$, then there exists $\epsilon>0$ such that the sequence AAA with prices $a-2b-2\epsilon, b+\epsilon, b+\epsilon$ implements a CE''.
It is easy to find such an $\epsilon$ by solving a set of linear inequalities.%
\footnote{
The prices are clearly not unique: any set of prices $p_1,p_2,p_3$ that are strictly larger than $b$ and have a sum of at most $a$ are CE prices.
}

As a second example, consider the sequence $A A B$. (H1) implies that the last price is $b$, which is by assumption larger than $c$, so (H3) is satisfied too. (H2) implies that the second price should be more than $b$ so set it to $b^+$; (H1) implies that the first price should be $(a-b)^-$, and (H2) then implies that $a-b \geq b^{++}$. To summarize:
\begin{align*}
\pixep{a > 2 b}
{\pick{A}{(a-b)^{-}}~~~\pick{A}{b^+}~~~\pick{B}{b}}
\end{align*}
This pixep implements a CE whenever $a>2b$. 
To see this, consider Alice first. Her only dominating bundle is the set of all three items; as claimed by Lemma \ref{lem:dominating}, Alice cannot afford it. She has many dominated bundles, for example, one containing Bob's item and one of her own items; by Lemma \ref{lem:dominated}, Alice does not prefer any of these bundles to her own. There are no unrelated bundles, so the allocation is a CE for Alice. 
For Bob the situation is opposite: his only dominated bundle is the empty bundle, which he of course does not want. He has many dominating bundles, which by Lemma \ref{lem:dominating} he cannot afford. He has no unrelated bundles, so the allocation is CE for Bob too.

\begin{remark}
When $a > 3b$, the two above pixeps, based on $AAA$ and $AAB$, both implement a CE. 
This raises the question  which CE is ``better''. 
Intuitively, if Alice's income is much larger than Bob's, it seems fairer to give all items to Alice, while if the difference is not so high, it seems fairer to leave the last item to Bob. However, these intuitions are not supported by the definition of CE. 
From the point-of-view of CE, which is the one taken in this paper, both pixeps are equally reasonable.
Moreover, the allocations yielded by both pixeps satisfy all the fairness properties described in Section \ref{sec:fairness}.
The question of selecting a single CE remains for future work.
\end{remark}

As a third example, consider the sequence $A B C$. (H1) implies that the prices must be $a, b, c$, and by assumption $a>b>c>$ all other incomes, so (H2,H3) are satisfied. Lemma \ref{lem:dominated} is applicable for all three agents, so to verify CE it is sufficient to consider unrelated bundles. 
Carl and Bob do not have unrelated bundles: for Carl, all bundles are dominating, except the empty bundle which is dominated.
For Bob, all bundles are dominating, except the empty bundle and Carl's bundle which are dominated.
For Alice, the empty bundle and all singleton bundles are dominated, the bundle of all items is dominating, and the pairs containing her own item are dominating too. Hence, Alice has a single unrelated bundle --- the one containing items \#2 and \#3. To ensure that Alice cannot afford that bundle, it is sufficient to require that $a < b+c$:
\begin{align*}
\pixep{a < b + c}
{\pick{A}{a}~~~\pick{B}{b}~~~\pick{C}{c}}
\end{align*}

Finally, consider the sequence $A B A$. To satisfy (H1,H3), set the prices to $(a-c)^-, b, c^+$. (H2) then implies that $(a-c)^- > b$, so:
\begin{align*}
\pixep{a > b + c}
{\pick{A}{(a-c)^-}~~~\pick{B}{b}~~~\pick{A}{c^+}}
\end{align*}
Here no agent has unrelated bundles. Lemma \ref{lem:dominated} implies that Bob does not want a dominated bundle.
It remains to verify that Alice does not want a dominated bundle of size $2$. 
The only such bundle is the pair of the items picked last. 
Suppose w.l.o.g. that Bob's ranking of singletons is: $x \succ y \succ z$.
Then Bob never picks $z$ so Alice gets either $xz$ or $yz$. If for Alice $xz\succ yz$ then she certainly picks $x$ first, so she prefers her bundle over the dominated bundle $yz$.
If for Alice $yz\succ xz$ then she has two options: pick $y$ first and get $z$ last, or pick $z$ first and get $y$ last. Both options lead to the same final allocation. In the first option she prefers her bundle to the dominated pair $xz$, while in the second option she might prefer the dominated pair $xy$.
However, 
to prove the {existence} of a CE, it is sufficient to prove that \emph{there exists} a SPE in which the allocation satisfies the CE conditions, so it can be assumed that Alice picks the first option.%
\footnote{
Currently I do not have an algorithm for finding the SPE that leads to a CE allocation.
When the preferences are strict, 
the choice of SPE does not affect the allocation; see Appendix \ref{sec:spe}.
}

The conditions of the last two pixeps, $ABC$ and $ABA$, cover all the income space except the hyperplane $a = b+c$, which has a measure of zero in $\allincomes$. Thus, the following theorem is proved:
\begin{theorem}
\label{thm:3-items}
When there are $m=3$ goods, 
for almost all incomes, 
for all monotonically\-/increasing preference\-/profiles, 
a CE exists and can be found by Algorithm \ref{alg:3 items}.
\end{theorem}

\begin{algorithm}
\begin{align*}
\pixep{a > b+c}
{\pick{A}{(a-c)^-}~~~\pick{B}{b}~~~~~~~~~\pick{A}{c^+}}
\\[3mm]
\pixep{a < b+c}
{\pick{A}{a}~~~~~~~~~\pick{B}{b}~~~~~~~~~\pick{C}{c}}
\end{align*}
\caption{
\small
\label{alg:3 items}
Competitive Equilibrium with $m=3$ items.
The algorithm works in almost all the income space, i.e, for all income-vectors $(a,b,c,\dots)$ in which $a>b>c>\ldots$ and $a+b\neq c$.
}
\end{algorithm}


\section{Four Goods}
\label{sec:4items}
\subsection{Two agents}
\label{sub:4items-2agents}

\begin{algorithm}
\begin{enumerate}
\item If ~~ ${a > 2 b}$ ~~ then ~~ 
${\pick{A}{(a-b)^{--}}~~~\pick{A}{b^+}~~~\pick{B}{b}~~~\pick{A}{0^+}}$~~~~.
\\[3mm]
\item If ~~ ${a < 2 b}$ ~~ then ~~ 
play the sequential game below:
\begin{align*}
&\text{Alice may choose:} && {\pick{A}{a^{--}}~~~\pick{B}{b^{-}}~~~\pick{A}{0^{++}}~~~\pick{B}{0^+}}~~~~.
\\[3mm]
&\text{Else, Bob may choose:} && {\pick{B}{b}~~~\pick{A}{b^{-}}~~~\pick{A}{(a-b)^-}~~~\pick{A}{0^{++}}}~~~~.
\\[3mm]
&\text{Else:} && {\pick{A}{a/2}~~~\pick{A}{a/2}~~~\pick{B}{b/2}~~~\pick{B}{b/2}}~~~~.
\end{align*}
\end{enumerate}

\caption{
\small
\label{alg:4 items 2 agents}
Competitive Equilibrium with $m=4$ items and $n=2$ agents.
Works for all income-vectors $(a,b)$ with $a>b, a\neq 2 b$.
}
\end{algorithm}
In this section there are $m=4$ items. 
Initially we assume there are only two agents --- Alice and Bob --- with incomes $a>b$. 
\citet{babaioff2019competitive}~[Proposition 3.2] 
already proved that in this case CE exists in almost all income space, but the algorithm presented here is shorter.

Algorithm \ref{alg:4 items 2 agents} uses a higher-level sequential game --- a game in which agents may choose between different pixeps. 
Definition \ref{def:pixep-finds-ce} naturally extends to this higher-level game: 
it implements a CE given income-vector $\incomevector$ if whenever the income-vector of the agents is $\incomevector$, for \emph{every} preference-profile in $\allpreferences$, there \emph{exists} a SPE $Q$ of this game such that the resulting allocation with the resulting price-vector are a CE.

\begin{theorem}
\label{thm:2-agents-4-goods}
When there are $n=2$ agents and $m=4$ goods, 
for almost all incomes,
for all monotonically\-/increasing preference\-/profiles, 
a CE exists 
and can be found by Algorithm \ref{alg:4 items 2 agents}.
\end{theorem}
\begin{proof}
The case $a > 2b$ is handled by the sequence $AABA$:
\begin{align*}
\pixep{a > 2 b}
{\pick{A}{(a-b)^{--}}~~~\pick{A}{b^+}~~~\pick{B}{b}~~~\pick{A}{0^+}}
\end{align*}
All three requirements on the price-sequence (H1,H2,H3) are clearly satisfied. No agent has any unrelated bundles. By lemma \ref{lem:dominated}, Bob does not want a dominated bundle.
It only remains to check that Alice does not want a dominated bundle of size 3. This can be verified similarly to the case $A B A$ in the previous section. There exists a SPE in which Bob's worst item is picked (by Alice) at the last step. Alice receives the best of the three triplets that contain this item, so she prefers it to any dominated triplet.

The case $a < 2b$ is more complicated. 
It requires letting agents choose between different pixeps. This leads to the following  three-step sequential game.

Step \#1: Alice may choose the following pixep based on $ABAB$:
\begin{align*}
{\pick{A}{a^{--}}~~~\pick{B}{b^{-}}~~~\pick{A}{0^{++}}~~~\pick{B}{0^+}}~~~~.
\end{align*}

Step \#2: If Alice does not choose $ABAB$, then Bob may choose the following pixep based on $BAAA$:
\begin{align*}
{\pick{B}{b}~~~\pick{A}{b^{-}}~~~\pick{A}{(a-b)^{-}}~~~\pick{A}{0^{++}}}~~~~.
\end{align*}

Step \#3: If Bob does not choose $BAAA$, then we play the following pixep based on $AABB$:
\begin{align*}
{\pick{A}{a/2}~~~\pick{A}{a/2}~~~\pick{B}{b/2}~~~\pick{B}{b/2}}~~~~.
\end{align*}
To understand the design of this algorithm, it is useful to think of a special case in which the preferences are additive and identical; then, a bundle given to one agent should be more expensive than a bundle given to the other agent iff it is more valuable. Suppose $w\succ x\succ y \succ z$, so the items are picked in the order $w,x,y,z$. Then step 1 handles the case $w\succ xz$, step 2 handles the case $xz\succ w\succ yz$, and step 3 handles the remaining case $yz \succ w$.

Proving that the same algorithm works for general preferences requires analysis using backward induction. 
Rename the items such that Bob's best item is $w$, and for Alice: $wx \succ wy \succ wz$.

In Step 3, Alice gets her best pair and Bob gets its complement. By Lemma \ref{lem:dominated} no agent wants a dominated bundle. 
Alice's unrelated bundles are all triplets, and Alice cannot afford even the cheapest of them, since $a/2+b > a$.
Bob's unrelated bundles are all singletons, 
and Bob can afford them. But, if he wanted a singleton, he could choose $BAAA$ in the previous step and get his best singleton ($w$). So in step 3 it is safe to assume that Bob does not want an unrelated bundle.

In Step 2, Bob has to choose between $w$ (his preferred singleton), and the complement to Alice's best pair. There are three cases:
(a) If Alice's best pair is $xy$ or $xz$ or $yz$, then the complement contains $w$ so Bob certainly prefers $AABB$. 
(b) Otherwise, Alice's best pair is $wx$ and its complement is $yz$; if for Bob $yz\succ w$, then again he prefers $AABB$.
(c) Only if Alice's best pair is $wx$ and for Bob $w\succ yz$, does Bob choose $BAAA$.
In the latter case, Bob's bundle is $w$. There exists a SPE in which Alice chooses her three items in the order: $x,y,z$. Then, Bob cannot afford $xy$ or $xz$ since they cost more than $b$. The only unrelated bundle he can afford is $yz$. However, in case (c) Bob prefers $w$ to $yz$, so Bob does not want any unrelated bundle.
Alice can afford only three unrelated bundles --- $w$ and $wy$ and $wz$. However, if she wants any of these, she could choose in the previous step $ABAB$ and pick $w$ first; this would guarantee her at least $wy$. So there exists a SPE in which, in step 2, Alice does not want any unrelated bundle that she can afford.

In Step 1, in cases (a-b) above, Alice never chooses $ABAB$, since she can get her best pair by waiting for Step 3.  In case (c), Alice chooses $ABAB$ iff she prefers the pair she is going to get over the triplet $xyz$; this pair must contain $w$, so it can be  assumed that if $ABAB$ is played, Alice picks $w$ first.%
\footnote{
If Alice picks $w$ first, then her bundle is either $w y$ (if Bob picks $x$ first) or $w x$ (otherwise).
If Alice picks $y$ or $z$ first, then her bundle cannot be better for her than $w y$, so we can assume she does not do this.
If Alice picks $x$ first, then Bob picks $w$ second, because in case (c) he prefers $w$ to $y z$; but 
then Alice's bundle is contained in $xyz$, so it is worse for her than $w y$.
}
Now, Bob has only one unrelated bundle $w$, which he cannot afford since $a>b$.
Alice has one unrelated bundle $xyz$, which by assumption she does not want.

It remains to check the dominated bundles in the case $ABAB$.
Alice receives a pair that she prefers over $xyz$, so she prefers it to every pair contained in $xyz$.
Her only other dominated pair contains $w$ and the item picked last. But if Alice wanted this pair, she could have picked this item in her before-last turn.
From Bob's point of view, the relevant sequence is $BAB$, which is analogous to the sequence $ABA$ analyzed in the previous section. Therefore, Bob too does not want any dominated pair.
\qed
\end{proof}

\subsection{Three agents with general monotonically-increasing preferences}
\label{sub:4items-3agents}
In this subsection there are four items and three agents --- Alice Bob and Carl --- with incomes $a>b>c$.

In the conference version of this paper \citep{SegalHalevi2018CEFAI}, 
I presented an algorithm using pixeps
for this case, and claimed that it implements a CE for almost all incomes.
This was a mistake: the algorithm works only when the agents have additive preferences (see subsection 
\ref{sub:4items-3agents-additive} below).
The error was in one pixep that was supposed to handle a specific range of incomes;
by focusing on this specific range, I  found a specific preference-profile for which no CE exists.%
\footnote{
To decrease the chances
that this example, too, has a mistake,
I wrote a Python program
that, given preferences and incomes,
exhaustively checks all allocations, and for each allocation, looks for CE prices using linear programming. 
I used the program to verify this and the other non-existence results in this paper.
The code is available here:
https://github.com/erelsgl/indivisible-competitive-equilibrium
}

\begin{theorem}
\label{thm:3-agents-4-goods}
When there are $n= 3$ agents and $m\geq 4$ goods, 
there is a monotonically-increasing preference-profile and a positive-measure subset of the income-space, where no allocation is CE-fair, hence no CE exists.
\end{theorem}
\begin{proof}
Consider first the case of $m=4$ goods.
Call the agents Alice Bob and Carl and denote their incomes by $a,b,c$. 
Consider the income subspace defined by the following  inequalities:
\begin{align*}
a>b>c ~~~ \text{and} ~~~ 
4c > a+b > 2b > a > b+c > 2c
\end{align*}
It has a positive measure since it is open and contains e.g. the point $(20, 11, 8)$.

There are four items denoted by: $w, x, y, z$.
The agents' preferences contain the following relations (Carl's preferences are irrelevant for the proof):
\begin{itemize}
\item     Alice:   $\text{triplets} \succ wx \succ yz\succ wy \succ xz\succ wz \succ xy  \succ \text{singletons} \succ \emptyset$.
\item     Bob:    $w \succ x \succ y \succ z \succ \emptyset$.
\end{itemize}

Suppose by contradiction that a CE-fair allocation exists. Applying \eqref{eq:1} to Carl gives:
\begin{align*}
c > {1\over 4}(a+b)
&&\implies&&
C\succeq \maxmin{\big(A\cup B\big)}{1}{4}
\end{align*}
so Carl must get at least one item.

CE-fairness implies that Bob must not envy Carl, so he must get at least one item too.

Applying \eqref{eq:1} to Alice gives:
\begin{align*}
a > b+c
&&\implies&&
A\succeq B\cup C
\end{align*}
so Alice must get at least two items, and these must be $wx$ or $wy$ or $wz$ (since Alice prefers these three pairs to their complements).
So Alice must get exactly 2 items, Bob exactly 1 and Carl exactly 1.

If Alice gets $wx$, then Bob gets $y$ or $z$. In both cases 
\eqref{eq:1} is violated for him, since:
\begin{align*}
b > a/2
&& \text{while} && 
B \prec \maxmin{A}{1}{2} = x.
\end{align*}

If Alice gets $wy$ or $wz$, then \eqref{eq:1} is violated for her, since:
\begin{align*}
a > {1\over 2}(a+b+c)
&& \text{while} && 
A \prec \maxmin{(A\cup B\cup C)}{1}{2} = \min(wx,yz)=yz.
\end{align*}
\qed
\end{proof}

Consider now the case $m>4$.
Intuitively, more items can only make it harder to find a CE.  The following lemma formalizes this intuition.
\begin{lemma}
\label{lem:more-items}
Let $n\geq 2$ and $m\geq 2$ be integers and $\incomevector\in \allincomes$ an income-vector.

(a) If there exists a CE w.r.t. $\incomevector$ for all preference-profiles in $\allpreferencesplusm$%
,
then there exists a CE w.r.t. $\incomevector$ for all preference-profiles in $\allpreferences$%
.

(b) If there exists a CE-fair allocation w.r.t. $\incomevector$ for all preference-profiles in $\allpreferencesplusm$%
,
then there exists a CE-fair allocation w.r.t. $\incomevector$ for all preference-profiles in $\allpreferences$%
 .
\end{lemma}
\begin{proof}
Given a preference-profile with $m$ items, construct a preference-profile with $m+1$ items by adding a ``low-value item'', denoted by $z$.
``Low-value'' means that, 
for every agent $i\in\allagents$,
and for every two bundles $X$ and $Y$ which do not contain $z$,
$Y \succeq_i X\cup \{z\}$
in the new profile
if and only if $Y\succeq_i X$
in the original profile.
In other words, adding the new item $z$ does not affect the relative order between bundles in the original profile.

\textbf{Part (a).}
By assumption, there exists a CE in the new profile. Denote it by $(\pricevector,\mathbf{X})$.
Suppose w.l.o.g. that in this CE item $z$ is given to agent $1$. Define a CE $(\pricevector,\mathbf{Y})$ in the original profile, where $\mathbf{Y}$ is identical to $\mathbf{X}$ except that $z$ is removed, i.e.:
$Y_j := X_j$ for all $j\neq 1$, and $Y_1 := X_1\setminus \{z\}$.
Then $(\pricevector,\mathbf{Y})$ is a CE: All agents except $1$ have the same bundle and see the same allocation and prices, so both CE conditions hold for them.
As for agent 1, 
the first CE condition holds since $p(Y_1) = p(X_1)-p(z) \leq p(X_1)\leq t_1$.
To check the second CE condition,
let $W$ be a bundle with 
$W \succ_1 Y_1$.
By the definition of the low-value item $z$, 
this implies that 
$W \succ_1 Y_1 \cup \{z\} = X_1$.
Since $(\pricevector,\mathbf{X})$ is a CE in the new profile, this 
implies that $p(W) > t_1$.

\textbf{Part (b)}.
By assumption, there exists a CE-fair allocation $\mathbf{X}$ in the new profile. Let $\mathbf{Y}$ be an allocation derived from $\mathbf{X}$ by removing $z$ from the bundle of the agent who received it. Then $\mathbf{Y}$  is CE-fair:
in both  \eqref{eq:1} and \eqref{eq:2}, the left-hand sides (the conditions on the incomes) do not change since the income-vector is the same. 
In the right-hand sides, removing $z$ from either side of the inequality does not affect the maximin-share bundles and the order relation between the bundles, since $z$ is a low-value item. 
\qed
\end{proof}

Combining Lemma \ref{lem:more-items} with the proof for $m=4$ completes the proof of Theorem \ref{thm:3-agents-4-goods} for any $m\geq 4$.

\subsection{Three agents with additive preferences}
\label{sub:4items-3agents-additive}
This subsection presents an algorithm, based on pixeps, 
for finding a CE for almost all incomes,
when there are $m=4$ goods and $n=3$ agents with additive preferences. 
Its proof of correctness uses the following lemma, which extends Lemma \ref{lem:dominated}.
\begin{lemma}
\label{lem:dominated-gap}
Let $i$ be an agent with additive preferences.
Suppose in a pixep $S$,
the turns of agent $i$ are arranged in a contiguous sequence, after it a single turn of another agent $j$, and then another turn of agent $i$. 
Then there exists a SPE $(S,Q)$ in which agent $i$ does not want any dominated bundle.
\end{lemma}
\begin{proof}
Suppose the contiguous sequence is of length $k-1$, so all in all agent $i$ is entitled to $k$ items. 
Let's say that agent $i$ \emph{picks in order}, if in each turn he picks his most valuable remaining item.
By additivity, picking in order guarantees the agent at least his 2nd-best $k$-tuple from the set of items remaining on the table when his first turn arrives. Therefore only the two following cases are possible:

Case \#1: Agent $i$ has a strategy by which, in SPE, he gets the best $k$-tuple of remaining items. 
Then there exists a SPE in which agent $i$ plays this strategy. Since any dominated bundle contains at most $k$ remaining items, agent $i$ does not want any such bundle.

Case \#2: Agent $i$ does not have such a strategy. 
Then, there exists a SPE in which he picks in order and gets the 2nd-best $k$-tuple of remaining items.
When picking in order, all dominated bundles have a smaller value for agent $i$.
 \qed
\end{proof}

\begin{remark}
In Lemma \ref{lem:dominated-gap}
there is a ``gap'' of $1$ in the agent's turns.
The lemma no longer holds with a gap of $2$.
For example, consider three agents with additive preferences satisfying:
\begin{itemize}
\item     Alice:  $w \succ x \succ y \succ z$,
\item     Bob:~~  $x \succ w \succ y \succ z$,
\item     Carl:~  $y \succ z\succ x\succ w$.
\end{itemize}
Consider the picking-sequence $ABCA$. 
There is only one SPE: the items are picked in the order $y x z w$, and Alice's bundle is $w y$
(if Alice picks any other item instead of $y$ in her first turn, she gets $w z$ or $x z$, both of which are worse for her than $w y$). 
Now Alice prefers the dominated bundle $w x$. \qed
\end{remark}

\begin{theorem}
\label{thm:3-agents-4-goods-additive}
When there are $m=4$ goods and $n=3$ agents,
for almost all incomes,
for all \emph{additive monotonically-increasing} preference\-/profiles, 
a CE exists 
and can be found by Algorithm \ref{alg:4 items 3 agents}.
\end{theorem}
\begin{algorithm}
\begin{enumerate}
\item If $a > 2 b + c$ \hskip 1mm then \hskip 1mm
${
\pick{A}{(a-b-c)^{--}}
\pick{A}{b^+}~~~\pick{B}{b}~~~\pick{A}{c^+}}$~~~~.
\vskip 5mm
\item If ~~ $2 b + c > a > 2 b$ ~~ then ~~~~~~~~~~ 
${\pick{A}{(a-b)^{-}}~~~\pick{A}{b^+}~~~\pick{B}{b}~~~\pick{C}{c}}$~~~~.
\vskip 5mm
\item {\scriptsize If  $2 b > a > b + c ~\&~ a + c > 2 b$ then }
${\pick{A}{b^{+}}~~~\pick{B}{b}~~~\pick{A}{(a-b)^{-}}~~~\pick{C}{c}}$~~~~.
\vskip 5mm
\item If ~~ $2 b > a > b + c$ ~ and ~ $2 b > a + c$ (implies $b > 2 c, a> 3 c$) ~~ then:
\begin{align*}
&\text{Alice may choose:} && {\pick{A}{(a-c)^{--}}~~~\pick{B}{(b-c)^{-}}~~~\pick{A}{c^{++}}~~~\pick{B}{c^+}}~~~~.
\\~\\
&\text{Else, Bob may choose:} && {\pick{B}{b}~~~\pick{A}{(a-2 p)^{--}}~~~\pick{A}{p^+}~~~\pick{A}{p^+}}~~~~.
\\
&  && \text{\color{gray} where $p:=\max~(c,~(a-b)/2)$}
\\~\\
&\text{Else:} && {\pick{A}{a/2}~~~\pick{A}{a/2}~~~\pick{B}{b/2}~~~\pick{B}{b/2}}~~~~.
\end{align*}
\vskip 5mm
\item If ~~ $b + c > a > 2 c$ ~and~ $2 c > b$ ~~ then ~~   play:
\begin{align*}
& \text{Alice may choose:} &&
{\pick{A}{a}~~~\pick{B}{b^-}~~~\pick{C}{c}~~~\pick{B}{0^+}}~~~~.
\\[3mm]
& \text{Else:} &&
{\pick{B}{b}~~~\pick{A}{(a-c)^-}~~~\pick{A}{c^+}~~~\pick{C}{c}}~~~~.
\end{align*}
\vskip 5mm
\item If ~~ $b + c > a > 2 c$ and $b > 2 c$ ~~ then ~~ play:
\begin{align*}
&\text{Alice may choose:} &&
{\pick{A}{a}~~~\pick{B}{b/2}~~~\pick{B}{b/2}~~~\pick{C}{c}}~~~~.
\\~\\
&\text{Else, Bob may choose:} && {\pick{A}{(a-c)^{--}}~~~\pick{B}{(b-c)^{-}}~~~\pick{A}{c^{++}}~~~\pick{B}{c^+}}~~~~.
\\~\\
&\text{Else:} && 
 {\pick{B}{b}~~~\pick{A}{(a-c)^{-}}~~~\pick{A}{c^+}~~~\pick{C}{c}}~~~~.
\end{align*}
\vskip 5mm
\item If ~~ $2 c > a$ ~~ then ~~   play the sequential game below:
\begin{align*}
& \text{Alice may choose:} &&
{\pick{A}{a}~~~\pick{B}{b^-}~~~\pick{C}{c}~~~\pick{B}{0^+}}~~~~.
\\~\\
& \text{Else:} &&
{\pick{B}{b}~~~\pick{A}{c^{+}}~~~\pick{C}{c}~~~\pick{A}{(a-c)^-}}~~~~.
\end{align*}
\end{enumerate}
\caption{
\small
\label{alg:4 items 3 agents}
Implementing Competitive Equilibrium with $m=4$ items and $n=3$ additive agents with incomes $a>b>c$.
}
\end{algorithm}
\begin{proof}
First, it is easy to check that all pixeps satisfy (H1,H2,H3); in particular, all price-sequences are decreasing.
Moreover, for all agents, the turns are either contiguous or contiguous with a single gap. 
Hence, by Lemmata \ref{lem:dominated} and \ref{lem:dominated-gap}, no agent wants a dominated bundle.
Moreover, Carl has no unrelated bundles.
It only remains to check the unrelated bundles of Alice and Bob.

\textbf{Ranges 1 and 2 and 3} are straightforward: no agent can afford an unrelated bundle.%
\footnote{
The algorithm fails in 
range 3 
when the preferences are not additive.
For example, in the scenario used to prove Theorem \ref{thm:3-agents-4-goods},
in SPE 
Alice picks $w$, Bob picks $x$, Alice picks $y$ and Carl picks $z$, and Alice prefers the dominated bundle $yz$.
}

\textbf{Range 4} is analyzed similarly to range 2 in Algorithm \ref{alg:4 items 2 agents}. 
Note that the picking-sequences in both ranges are the same --- only the prices are different.
This means that the strategic behavior of the agents is exactly the same in both ranges.
The reason is that, in the Fisher market model, 
money is used only to purchase items in the market, and has no value outside the market.
Therefore, the preferences of the agents depend only on the bundles that they receive, and not on the prices.

\textbf{In Range 5}, 
In both steps, Bob cannot afford any unrelated bundle. 
To analyze Alice's strategy, rename the items such that Bob's best item is $w$ and Alice's best pair without $w$ is $xy$. Then Alice chooses $ABCB$ iff she prefers $w$ to $xy$. 

If she chooses $ABCB$, she gets $w$, does not want $xy$ (or any other pair without $w$), and cannot afford a triplet.
 
If she chooses $BAAC$, she gets $xy$, does not want $w$, and cannot afford any unrelated pair.

To analyze \textbf{Range 6}, rename the items such that
Bob's best item is $w$, 
Alice's best items besides $w$ are $x,y$,
and for Bob $x>y$.

In the last step $BAAC$, 
Bob picks $w$ and Alice picks $xy$ and the allocation is $(xy,w,z)$.
Alice can afford only one unrelated bundle --- the singleton $w$.
But if she wanted it, she could have chosen it in the first step $ABBC$.
Bob can afford only one unrelated bundle --- the cheaper of $xz,yz$. There exists a SPE where Alice picks $x$ before $y$; then, Bob can afford only $yz$. 
However, we get to $BAAC$ only if Bob prefers $w$ to the pair he could get in $ABAB$, which must be one of $xy,xz,yz$ and not the worst one. 
So by choosing $ABAB$ Bob could get at least $yz$. But he chose otherwise, hence he prefers $w$.

In the second step $ABAB$, 
Alice cannot afford even her cheapest unrelated bundle (triplet) since $b+c>a$.
Bob can afford only one unrelated bundle --- a singleton.
But Bob chooses $ABAB$ only if he prefers the pair he is going to get to his best singleton $w$.

For analyzing the first step $ABBC$, 
note that in $ABBC$ Alice gets a singleton, while in $ABAB$ Alice can get the same singleton plus another item.
Therefore, Alice chooses $ABBC$ only if she knows that Bob would choose $BAAC$.
By choosing $ABBC$, she indicates that she prefers her best item over the best pair that does not contain this item.
Therefore, she does not want any unrelated pair. Additionally, she cannot afford any triplet. 
Bob cannot afford even his cheapest unrelated bundle (singleton).

\textbf{In Range 7}, in both steps, Bob cannot afford any unrelated bundle. 

To analyze Alice's strategy, 
rename the items such that Bob's best item is $w$, and Carl's worst item besides $w$ is $z$.

In the last step, Bob picks $w$ and Carl picks $x$ or $y$, so Alice can get the best of $xz,yz$.
She can afford only one unrelated singleton ($w$), but if she wants it she can choose $ABCB$ in the first step.

In the first step necessarily Alice prefers $w$ to $best(xz,yz)$, so $w$ is her best singleton and she picks it.
Since $b+c>a$, 
Alice can afford only two unrelated pairs --- the ones containing the last item.
There exists a SPE in which this last item is $z$. So Alice can afford only $xz,yz$.
But if Alice wanted one of these, she could have waited to the last step.

Finally, it is easy to check that the seven ranges handled by Algorithm \ref{alg:4 items 3 agents} cover all the income space except a finite number of hyperplanes (corresponding to the equalities $a=b,b=c,a=2b+c,a=2b,a=b+c,a+c=2b,a=2c,2c=b$). Therefore, a CE exists for almost all incomes. \qed
\end{proof}

\subsection{Four or more agents}
\label{sub:4items-4agents}

In this subsection there are four goods and four (or more) agents. 

Intuitively, 
one would expect an analogue of Lemma \ref{lem:more-items}, which would say that a CE is less likely to exist when there are more agents. Such a lemma, combined with the impossibility result of Theorem \ref{thm:3-agents-4-goods}, would imply 
impossibility for any $m\geq 4$ and $n\geq 4$. 

\ifdefined\LemmaMoreAgents
However, so far I could prove such a lemma
only for the (simple) special case in which there are less items than agents.
\begin{lemma}
\label{lem:more-agents}
Let $n\geq 2$ and $m\geq 2$ be integers \emph{with $m\leq n$}.

(a)
If a CE exists for $n+1$ agents and $m$ items,
for almost-all
income-vectors in $\allincomesplus$ and all preference-profiles in $\allpreferencesplusn$,
then a CE exists with $n$ agents
and $m$ items, for almost-all
income-vectors in $\allincomes$
and all preference-profiles in $\allpreferences$.

(b)
If a CE-fair allocation exists with $n+1$ agents, for almost-all
income-vectors in $\allincomesplus$
and all preference-profiles in $\allpreferencesplusn$,
then a CE-fair allocation exists with $n$ agents, for almost-all
income-vectors in $\allincomes$
and all preference-profiles in $\allpreferences$.
\end{lemma}
\begin{proof}

\textbf{(a).}
Let $T^{n+1}\subseteq \allincomesplus$ be the set of income-vectors for which a CE exists for all preference-profiles in $\allpreferencesplusn$. 

Let $T^{n}\subseteq \allincomes$ be the projection of $T^{n+1}$ on $\allincomes$ (derived by just removing the $(n+1)$-th element from every vector in $T^{n+1}$).

...

Given a preference-profile with $n$ agents
and income-vector $\incomevector\in\allincomes$,
create a preference-profile with $n+1$ agents by adding a ``low-income agent''.
...
low-value item. 
\qed
\end{proof}

When $m > n$, in some specific instances, a CE may be \emph{more} likely to exist when there are more agents --- see Appendix \ref{sec:more-agents}. 

\else 
However, so far I could not prove such a general lemma.
In fact, in some specific instances, a CE may be \emph{more} likely to exist when there are more agents --- see Appendix \ref{sec:more-agents}. 
\fi
Therefore, the non-existence for $n\geq 4$ agents must be proved explicitly. 
An additional benefit of the proof below is that it is valid also for additive agents.

\begin{theorem}
\label{thm:4-agents-4-goods-additive}
With $m\geq 4$ goods
and $n\geq 4$ agents, 
there exists an \emph{additive monotonically-increasing} preference-profile and a positive-measure subset of the income-space,
where no allocation is CE-fair (hence no CE exists).
\end{theorem}
\begin{proof}
By Lemma \ref{lem:more-items}, it is sufficient to prove the theorem for $m=4$ goods.

Call the agents with the four highest incomes Alice Bob Carl and Dana. Denote their incomes by $a,b,c,d$. 
Consider the income subspace defined by the following  inequalities:
\begin{align*}
   2 a > 2 b  > a+c  >   b + d  >  2 c  >   a   >  c + d  >  2 d  >  b  >  c  >  d > \text{all other incomes}
\end{align*}
It is open and contains e.g. the point $(160,130,90,66,\ldots)$, so it has a positive measure.
~~~There are four items denoted by $w, x, y, z$.
The agents' preferences contain the following relations:
\begin{itemize}
\item     Alice:   $xy \succ w \succ xz \succ yz \succ x \succ y \succ z \succ \emptyset$.
\item     Bob:    $\text{pairs}\succ w \succ z \succ x \succ y\succ \emptyset$.
\item     Carl:    $\text{pairs}\succ x \succ y \succ w \succ z\succ \emptyset$.
\end{itemize}
Note that these preferences are additive. For example, Alice's valuations for $w,x,y,z$ can be $11,7,5,3$,
Bob's valuations can be $9,7,6,8$ and 
Carl's valuations can be $7,9,8,6$.
The preferences of Dana, as well as of the $n-4$ agents with the lower incomes (if any), are irrelevant for the proof.

Suppose by contradiction that a CE-fair allocation exists. 
In any allocation, at most four agents can get non-empty bundles, and by CE-fairness \eqref{eq:1}, these must be the four agents with the highest incomes.
Denote their bundles by $A,B,C,D$.
Then, the CE-fairness properties imply the following.

\textbf{1. Dana gets no item. }
If she gets an item, then by envy-freeness (which follows from \eqref{eq:1}), the three higher-income agents must also get an item: Alice must get $w$ (her best item), Bob must get $z$ (his best remaining item), Carl must get $x$ (his best remaining item), and Dana gets $y$.
But now \eqref{eq:1} is violated for Alice, since for her:
\begin{align*}
a > c+d
&& \text{while} && 
A \prec C\cup D = xy.
\end{align*}

\textbf{2. Bob and Carl get at most one item.}
If any of them gets more than one item, then \eqref{eq:1} is violated for Dana, since she gets an empty bundle, and for her:
\begin{align*}
d > b/2 > c/2 
&&\text{while}&&
D\prec \maxmin{B}{1}{2} 
&&
\text{or}
&&
D\prec \maxmin{C}{1}{2}.
\end{align*}

\textbf{3. Bob and Carl get exactly one item.}
If one of them (say, Carl) gets no items, then Alice gets at least 3 items, and \eqref{eq:1} is violated for Carl, since:
\begin{align*}
c > {1\over 2}a 
&&\text{while}&&
C\prec \maxmin{A}{1}{2}.
\end{align*}

\textbf{4. Alice gets $w$ plus another item}.
If she does not get $w$, 
then some lower-income agent gets it, so by envy-freeness  \eqref{eq:1}, Alice must get a bundle that she prefers to $w$. Such a bundle must contain both $x$ and $y$. Hence Carl (who gets at most one item) can get either $w$ or $z$.
But now \eqref{eq:1} is violated for Carl, since for him:
\begin{align*}
c > {1\over 2}a 
&&\text{while}&&
C\prec \maxmin{A}{1}{2}.
\end{align*}

So Bob gets a single item, and it is worse for him than $w$. But now \eqref{eq:1} is violated for Bob, since for him:
\begin{align*}
b > {1\over 2}(a+c)
&&\text{while}&&
B\prec \maxmin{(A\cup C)}{1}{2}.
\end{align*}
This is because $A\cup C$ can be partitioned into $w$ plus a pair, and Bob prefers both of these parts to his single item.
\qed
\end{proof}

\section{Five or More Goods}
\label{sec:5items}
In this section there are five goods. 
\citet{babaioff2019competitive}~[Theorem 3.5]
already showed that, with two agents, there may exist a subset of the income-space with a positive measure in which no CE exists.
For completeness, the theorem below provides an alternative proof that
uses the same income range and almost the same preferences, but is based only on the CE-fairness properties 
\eqref{eq:1} and \eqref{eq:2}.
Thus, it shows that even a CE-fair allocation without the other properties of CE (such as weak Pareto-efficientity) may be unattainable in a subset of positive measure.
Moreover, the proof is extended to any number of agents.
\begin{theorem}
\label{thm:2-agents-5-items}
With $n \geq 2$ agents and $m \geq 5$ goods, 
there exist a monotonically\-/increasing preference\-/profile and a subset of the income-space with a positive measure, for which no CE-fair allocation exists (hence no CE exists).
\end{theorem}
\begin{proof}
By Lemma \ref{lem:more-items}, it is sufficient to prove the theorem for $m=5$ goods. 

Since there are $5$ goods, we can assume w.l.o.g. that there are at most $5$ agents, since if there are more agents, the $n-5$ lower-income agents necessarily get an empty bundle in any CE-fair allocation.%
\footnote{
If the second conjecture in Appendix \ref{sec:more-agents} were proved, then we could assume w.l.o.g. that there are $n=2$ agents.
}
Denote the two highest incomes by $a,b$,
and the other incomes by $c,d,e$ (if there are less than 5 agents, then $e=0$; if there are less than 4 agents, then in addition $d=0$; if there are less than $3$ agents, then in addition $c=0$).
Consider the income subspace defined by:
\begin{align*}
a > b+c+d+e
\\
b > 3(a+c+d+e)/4 
\end{align*}
There are five goods: $v, w, x, y, z$.
The preferences of Alice and Bob (the two highest-income agents) 
contain the following relations:
\begin{itemize}
\item     Alice:   $\text{quartets} \succ vwx,vwy,vwz \succ vw \succ xyz \succ
\\
vxy,vxz,vyz,wxy,wxz,wyz
\succ \text{pairs-except-}vw \succ \text{singletons}$
\item     Bob: 
$\text{quartets} \succ \text{triplets-except-}xyz \succ vx,vy,vz,wx,wy,wz \succ xyz \succ vw \succ w \succ v \succ xy,xz,yz \succ x,y,z$
\end{itemize}
The preferences of the other agents, if any, are irrelevant for the proof.

Suppose by contradiction that a CE-fair allocation exists. 
Denote the bundles of the two highest-income agents by $A,B$,
and the other bundles by $C,D,E$
(if there are less than 5 agents, then $E=\emptyset$; if there are less than 4 agents, then in addition $D=\emptyset$; if there are less than 3 agents, then in addition $C=\emptyset$).

Applying \eqref{eq:1} to Alice gives:
\begin{align*}
a > {1\over 2}(a+b+c+d+e)
&&\implies&&
A\succeq \maxmin{\big(A\cup B\cup C\cup D\cup E\big)}{1}{2} \succeq xyz
\end{align*}
(the rightmost bundle is at least $xyz$ by the 2-partition $(vw,xyz)$).
We now check Alice's possible bundles by order of Alice's preferences, from $xyz$ upwards.

If $A=xyz$, then $B\cup C\cup D\cup E =vw$ and 
\eqref{eq:1} is violated for Alice, since for her $a > b+c+d+e$ while $A \prec B \cup C\cup D\cup E$.

If $A = vw$, then $B \subseteq xyz$ and $B \cup C \cup D \cup E = xyz$. 
Now, 
\eqref{eq:2} is violated for Bob, since $2b/3 > (a+c+d+e)/2$, so $b > b/3 + (a+c+d+e)/2$, so:
\begin{align*}
b > {1\over 3}(b+c+d+e) + {1\over 2}a
&&\text{while}&&
B\prec \left(\maxmin{(B\cup C\cup D\cup E)}{1}{3}\right) \sqcup \left(\maxmin{A}{1}{2}\right) 
\end{align*}
(the rightmost bundle contains one item from $B\cup C\cup D\cup E$ and one item from $A$, so it is one of $vx,vy,vz,wx,wy,wz$;  
Bob prefers all these to $xyz$, which contains $B$).

If $A$ is $vwx$ or $vwy$ or $vwz$, then 
$B$ is contained in $yz$ or $xz$ or $xy$.
But then \eqref{eq:1} is violated for Bob, since for him:
\begin{align*}
b > {1\over 3}(a+b+c+d+e) 
&&\text{while}&&
B\prec \maxmin{\big(A\cup B\cup C\cup D\cup E\big)}{1}{3} 
\end{align*}
(the rightmost bundle is at least $v$ by the 3-partition $(xyz,v,w)$).

Finally, if Alice gets four or more items,
then \eqref{eq:1} is violated for Bob, since for him:
\begin{align*}
b > {3\over 4} a
&&\text{while}&&
B\prec \maxmin{A}{3}{4} 
\end{align*}
(the rightmost bundle contains a triplet, and Bob prefers all triplets to all singletons). 
\qed
\end{proof}

\section{Allocation of Chores}
\label{sec:chores}
So far it was assumed that all items are goods, or equivalently, that all preferences are monotonically-increasing. 
In this section it is assumed that all items are 
\emph{chores} (also called: bads), or equivalently, that all preferences are \emph{monotonically-decreasing} --- an agent always prefers a set to its supersets.

Following the economic literature \citep{Bogomolnaia2017Competitive}, 
in a chore-allocation problem,
each agent has a \emph{negative} income, so 
the set of all possible income-vectors for $n$ agents is $\incomevector\in\mathbb{R}^n_-$.

Similarly, the prices of all chores are {negative}. 
A CE is a pair $(\pricevector,\mathbf{X})$, where $\pricevector$ is a vector of $m$ negative numbers and $\mathbf{X}$ is an allocation. The CE conditions are defined exactly as for goods, namely:

\textbf{Condition 1.}
The price of an agent's bundle (which is a negative number) is at most the agent's income (which is a negative number too):
\begin{align*}
\forall i\in\allagents:~~~ p(X_i)\leq t_i
\end{align*}

\textbf{Condition 2.}
Each agent's bundle is the best bundle he/she can afford with his/her income.
Formally, for every bundle $Y\subseteq \allitems$:
\begin{align*}
\forall i\in\allagents:~~~
p(Y) \leq t_i 
~~~ \implies~~~ 
Y_i \preceq_i X_i
\end{align*}

Condition 1 implies that each agent must do chores whose total price drops below his income. This means that agents with a smaller (= more negative) income need to do more chores, or less desirable chores. 
Thus, CE can be seen as a rule for dividing chores fairly among agents with different liabilities, where an agent with liability $|t_i|$ (a positive number) has an income of $t_i$ (a negative number).

One implication of Condition 1 for chores (which does not hold for goods) is that each agent must do at least one chore. This means that a CE does not exist if there are less chores than agents ($m<n$). 

Properties \eqref{eq:1} and  \eqref{eq:2} do not use the sign of the income or prices. Therefore, Propositions \ref{prop:1} and \ref{prop:2} hold equally well with chores, and any CE allocation is CE-fair. In particular, a CE with equal incomes is still envy-free and satisfies the 1-out-of-$n$ maximin-share guarantee. 

\subsection{Two agents}
\label{sub:chores-2agents}
Intuitively, with two agents, allocating chores is equivalent to allocating exemptions from chores.
An exemption from chore is a good; therefore, chores-allocation is equivalent to goods-allocation.%
\footnote{
\citet{Bogomolnaia2017Competitive} explain why this equivalence does not work with $n\geq 3$ agents:
for each chore, there are $n-1$ identical exemptions to share. When $n-1\geq 2$, a CE might allocate two or more exemptions to the same agent, which is not a valid allocation.
}
This subsection proves this intuition formally.

Given a preference relation $\succeq_i$, define its \emph{dual preference relation}, $\succeq^*_i$, as a relation that satisfies, for every two bundles $X,Y$:
\begin{align*}
X \succeq^*_i Y
\iff 
(\allitems \setminus X) \succeq_i (\allitems \setminus Y)
\end{align*}
Note that $\succeq_i$ is monotonically-increasing iff $\succeq^*_i$ is monotonically-decreasing, so $\succeq_i$ represents preferences on goods iff $\succeq^*_i$ represents preferences on chores.
\begin{theorem}
\label{thm:duality}
Consider the following two allocation problems, defined on the same set $\allitems$ with $|\allitems|\geq 2$:
\begin{itemize}
\item \emph{Goods-allocation problem}, defined as allocating the items in $\allitems$ between two agents with monotonically-increasing preferences $\succeq_A,\succeq_B$ and positive incomes $(a,b)$;
\item \emph{Chores-allocation problem}, defined as allocating the items in $\allitems$ between two agents with 
monotonically-decreasing preferences $\succeq^*_A,\succeq^*_B$ and negative incomes $(-b,-a)$.
\end{itemize}
Then, there exists a CE in the goods-allocation problem, if-and-only-if
there exists a CE in the chores-allocation problem.
\end{theorem}

\begin{proof}
Throughout the proof,  assume w.l.o.g. that $a\geq b$.

\textbf{Goods CE $\implies$ Chores CE.}
Let $(\pricevector, (A,B))$ be a CE in the goods-allocation problem. Transform this CE in several steps.

\emph{Step 1. Ensure that both bundles are non-empty.} Alice's bundle $A$ obviously cannot be empty, since otherwise Alice envies Bob.
If Bob's bundle $B$ is empty (so $A=\allitems$), 
let $x\in \allitems$ be a good for which 
Alice prefers
$\allitems\setminus \{x\}$ to all bundles with $m-1$ goods. Then, set $A' := \allitems\setminus \{x\}$, $B' := \{x\}$, $p'(x)=b$, and $p'(y)=r\cdot p(y)$ for all $y\neq x$, where $r$ is a normalization constant defined by $r:= a / p(A')$.%
\footnote{
The expression for $r$ is well-defined thanks to the assumption $|\allitems|\geq 2$.
Without this assumption,
at least one agent must get an empty bundle of chores,
but then CE condition 1 is violated for this agent.
}
Note that $r\geq 1$, since $p(A')\leq p(A) \leq a$.

This $(\pricevector', (A',B'))$ is still a CE in the goods-allocation problem. \emph{Proof:} CE condition 1 holds, since $p'(A') = r p(A') = a$ and $p'(B') = b$.
CE condition 2 holds for Alice, since she prefers $A'$ to all bundles except $\allitems$, and she cannot afford $\allitems$.
CE condition 2 holds for Bob too, since in the original price-vector $\pricevector$, he could not afford any item. In the new price-vector $\pricevector'$, the prices of all items except $x$ are higher (they are multiplied by $r\geq 1$), so Bob certainly cannot afford them.

\emph{Step 2. Ensure that both agents exhaust their incomes.} Once both bundles are non-empty, if the price of an agent's bundle is less than the agent's income, arbitrarily increase prices of items in his/her bundle until the total bundle price equals his/her income. The CE conditions are still satisfied. To ease notation, denote the new CE by the same notation $(\pricevector', (A',B'))$.

\emph{Step 3. Construct a CE for the chore-allocation problem.} Allocate $B'$ to Alice, $A'$ to Bob, and set the price-vector to $-\pricevector'$. 

This $(-\pricevector', (B',A'))$ is a CE in the chore-allocation problem. \emph{Proof}: CE condition 1 obviously holds, since $-p'(B')=-b$ and $-p'(A')=-a$. To prove that CE condition 2 holds for Alice, it is sufficient to prove that every bundle that Alice prefers over $B'$ by her dual preference-relation, costs more than $-b$. Indeed, let $X$ be a bundle such that $X \succ_A^* B'$. By definition of the dual preference relation, 
$(\allitems \setminus X) \succ_A (\allitems \setminus B')$. But $\allitems \setminus B' = A'$, so $(\allitems \setminus X) \succ_A A'$.
Now, the bundle $\allitems \setminus X$ must be unaffordable for Alice in the goods-CE, so $a < p'(\allitems \setminus X)$. 
This implies $a < p'(\allitems) - p'(X)$.
But by Step 2 above, $p'(\allitems) = a+b$. Hence we get $a < (a+b) - p'(X)$ which implies $-p'(X) > -b$.
By an analogous argument, CE condition 2 holds for Bob.

\textbf{Chores CE $\implies$ Goods CE.}
Let $(-\mathbf{q}, (B,A))$ be a CE in the chores allocation problem. 
Normalize it as in Step 2 above, so that $-q(B)=-b$ and $-q(A)=-a$.
This $(\mathbf{q}, (A,B))$ is a CE in the goods allocation problem. 
The proof is entirely analogous to the proof in step 3 above.
\end{proof}

Theorem \ref{thm:duality} implies that the existence result of subsection \ref{sub:4items-2agents} and the non-existence result of section \ref{sec:5items} are true for chores too. So 
CE exists for almost all incomes
when there are two agents and at most four chores, 
but not when there are five or more chores.

\subsection{Three or more agents}
\label{sub:chores-3agents}
Unfortunately, the positive results for three agents do not hold for chores, as shown in the following theorem.

\begin{theorem}
With $n \geq 3$ agents and $m\geq 1$ chores,
there exists an additive monotonically\-/decreasing preference-profile and a positive-measure subset of the income-space, where no allocation is CE-fair (hence no CE exists).
\end{theorem}
\begin{proof}
Let $\incomevector\in\mathbb{R}^n_-$ be the vector of incomes, with:
\begin{align*}
2 t_1 < t_n < \cdots < t_1 < 0.
\end{align*}
There is one ``hard'' chore denoted by $z$. The agents' preferences are identical, monotonically-decreasing, and contain the following relations over the chores:
\begin{align*}
\emptyset \succ \text{[bundles without $z$]} \succ \{z\} \succ \text{[bundles with $z$]}
\end{align*}
Note that these preferences are additive. For example, the value of $z$ can be $-2^{2m}$ 
and the values of the other chores can be $-2, -4, -8, \ldots$.

Suppose by contradiction that there exists a CE-fair allocation $\mathbf{X}$.
Let $i\in\allagents$ be the agent who gets the hard chore $z$.
By assumption  $t_i > 2 t_1$, and $n\geq 3$ implies that $2 t_1 \geq \sum_{j\neq i} t_j$ (the left-hand side contains at least two negative terms, each of which is at most $t_1$).
But $i$ prefers all other chores over his bundle $X_i$, which contains the hard chore.  
So
\eqref{eq:1} is violated for $i$:
\begin{align*}
t_i >  \sum_{j\neq i} t_j
&& \text{while} && 
X_i \prec \bigcup_{j\neq i} X_j.
~~~~~~~~\qed
\end{align*}
\end{proof}

\section{Related Work}
\label{sec:related}
Babaioff, Nisan and Talgam-Cohen study several settings besides the ones mentioned in previous sections. 
First, they study the setting of two agents with additive preferences and an arbitrary number of goods \citep{babaioff2019fair}. They present several sufficient conditions by which a CE exists for almost all incomes, such as: identical preferences and arbitrary incomes, or arbitrary preferences and almost-equal incomes. The general case remains open. 
Second, they study other restricted preference domains besides additive preferences \citep{babaioff2019competitive}, such as the so-called \emph{lexicographic}, \emph{leveled}, \emph{responsive}, \emph{submodular}, and more. They present a hierarchy of preference-domains and present existence results for some domains in the hierarchy.

\subsection{CE with indivisible items}
A neat review of the different computational problems related to CE can be found in Appendix C of \citet{Babaioff2017Competitive}.

\citet{Deng2003Complexity} study an \emph{Arrow-Debreu market}. This is a generalization of the Fisher market studied in the present paper, in which each agent may be both a buyer and a seller, i.e., each agent may bring to the market an initial set of items to sell, rather than just an initial income.
They assume that all items are goods, and all agents have additive preferences.
They prove that deciding whether a CE exists is NP-hard even if there are $3$ agents. They present an approximation algorithm which relaxes the CE conditions in two ways: (1) The bundle allocated to each agent is valued at least $(1-\epsilon)$ of the optimum given the prices, and (2) at least a fraction $(1-\epsilon)$ of the items are allocated. Both these relaxations are unrelated to our setting, in which the preferences are ordinal and all items must be allocated.

\citet{bouveret2016characterizing} 
present a scale of five fairness criteria for allocating indivisible goods among agents with equal entitlements and additive preferences. 
The weakest criterion in their scale is the \emph{maximin share guarantee} (for all $i\in\allagents: X_i\succeq_i \maxmin{\allitems}{1}{n}$), and the strongest criterion is \emph{CE-from-equal-incomes (CEEI)}. They asked what is the computational complexity of deciding whether a CEEI exists. 
This question was answered soon afterwards by \citet{Aziz2015Competitive}, who proved that the problem is weakly NP-hard when there are two agents and $m$ items, and strongly NP-hard when there are $n$ agents and $3n$ items. 
\citet{Branzei2015Characterization} further proved that even verifying whether a given allocation is a CEEI is co-NP-hard. 

\citet{Branzei2015Characterization} study CEEI for the case in which the agents are \emph{single-minded}: for each agent $i$ there is a subset of goods $D_i\subseteq \allitems$, such that $i$ strictly prefers all bundles that contain $D_i$ to all bundles that do not contain $D_i$, is indifferent between any two bundles that contain $D_i$, and is indifferent between any two bundles that do not contain $D_i$.
In this case, verifying whether a given allocation is a CEEI is polynomial, but checking if a CEEI exists is co-NP-complete.
Single-minded agents are further studied by \citet{Branzei2016To}. In contrast to our setting, they assume that each good can come in multiple units. 
The main difference between $k$ units of the same good, and $k$ different goods that all agents consider to be equivalent, is that the $k$ units of the same good are required (by definition) to have the same price.
This restriction implies that the ``income-exhaustion'' requirement (H1) from subsection \ref{sub:pixeps} is no longer w.l.o.g.:
\citet{Branzei2016To} show an example in which (1) a CE where all agents exhaust their income does not exist, (2) a CE where some agents spend less than their income does exist. They call this solution CAEI --- Competitive Allocation from Equal Incomes. Interestingly, in contrast to a CEEI, it is possible to find a CAEI (if one exists) in polynomial time.

\citet{Heinen2015Fairness} 
extend the five-criteria scale of \citet{bouveret2016characterizing} 
for \emph{$k$-additive preferences}, in which each agent reports a value for each bundle of at most $k$ items, and the values of larger bundles are determined by adding and subtracting the values of the basic bundles. 
They study both goods and chores.
They show an example of an envy-free and Pareto-efficient allocation of goods, that is not a CE from equal incomes.

\citet{budish2011combinatorial} studies the most general setting in which agents can have arbitrary preference relations over bundles --- there is no monotonicity assumption --- the market can have a mixture of goods and chores.
He presents a beautiful and practical \emph{approximate CEEI} mechanism, which relaxes the CEEI conditions in two ways: (1) The agents' incomes are not exactly equal, and (2) a small number of items may have to be added or discarded (the motivating application is allocating course seats among students: usually it is possible to add a small number of seats to accommodate the demand). He proved that an approximate-CEEI always exists (although \citet{othman2016complexity} later proved that the computation of approximate-CEEI is PPAD-complete). 
The first relaxation (1) is closely related to the concept of CE for almost all incomes. 
In fact, both models can be described as adding an arbitrarily small random perturbation to a given income-vector; in his model, the given vector has equal incomes, and in our model, the given vector can have non-equal incomes (that correspond to agents with different entitlements). 
Budish's second relaxation (2) make his solution less useful when the number of items is fixed.

\citet{barman2019proximity} study Fisher markets of goods, in which all agents have additive preferences. They show that a fractional CE (where some goods are divided) can always be rounded to an integral CE (where goods remain indivisible), by changing the agents' incomes. The change in each income can be as high as the largest price of a good in the fractional CE. Thus, their paper does not answer the question whether a CE exists for almost all income-vectors.

\subsection{Picking-sequences}

Picking-sequences are common practical mechanisms for allocating indivisible items. 
They are favored in various real-life situations due to their simplicity, privacy and low communication complexity \citep{Bouveret2011General}. 

In matching markets, in which each agent is entitled to exactly one item, a common mechanism is \emph{random serial dictatorship} \citep{abdulkadiroglu1998random,manea2009asymptotic,aziz2013computational,bade2019random}.
This is a special case of a picking-sequence in which each agent has exactly one turn, and the sequence of turns is selected at random.
In this setting, each agent has a simple dominant strategy and it is to pick the best remaining item (as in Lemma \ref{lem:dominated}). 
In contrast, 
in a general picking-sequence, 
an agent who has several non-consecutive turns usually has no dominant strategies. This motivates various works that study the strategic behavior of agents in picking sequences.

\citet{Brams2004Dividing} and \citet{Brams2007Mathematics}
study picking-sequences for allocating cabinet ministries among parties. There is a coalition of parties; each party has a different number of seats in the parliament; larger parties should be allocated more ministries, or more prestigious ministries. This is an interesting use-case of fair division with different entitlements. A possible solution to this problem is to determine a picking-sequence, based on the different entitlements, and let each party pick a ministry in turn. Such a solution is used in Northern Ireland, Denmark and the European parliament \citep{OLeary2005Divisor}. 

Brams and Kaplan assume that all items are goods, each agent has a strict ordering on the items, and has \emph{responsive preferences} over bundles. 
This means that replacing an item in a bundle by a preferred item, always results in a preferred bundle. 
Responsive preferences are a strict superset of additive preferences and a strict subset of the monotonic preferences studied in this paper \citep{Babaioff2017Competitive}. With responsive preferences, at each point in the picking-sequence, there is a single remaining item which the picking agent most prefers. An agent can be ``truthful'' and pick this best item, or be ``strategic'' and pick another item based on his knowledge of the other agents' preferences.
\citet{Brams2004Dividing} prove several results that may be relevant for future development of pixeps: (1) With two agents, both truthful and strategic choices lead to Pareto-efficient allocations in SPE. Moreover, the game is monotonic in the following sense: 
if an agent's turns are moved earlier in the sequence, then the agent's bundle in SPE improves (e.g, Alice is better-off in an SPE of $ABBA$ than in an SPE of $BABA$). (2) Both properties are still true with three or more truthful agents. But with three or more strategic agents, a picking-sequence might have an SPE that is not Pareto-efficient. Moreover, the game might be non-monotonic, i.e, an agent who is allowed to pick earlier in the sequence, might have a worse bundle in SPE. (3) For two agents, a simple modification of the picking-sequence game leads to a mechanism in which picking truthfully is a dominant strategy (see \citet{kalinowski2013strategic} for a formal proof). This makes it easier to find pixeps that implement a CE, since an agent who picks truthfully never wants a dominated bundle.

\citet{kohler1971class} present a linear-time algorithm to compute a SPE in the special case in which there are $n=2$ agents, any number of goods,  and the picking-sequence is alternating ($ABABAB...$).
\citet{kalinowski2013strategic} generalize this algorithm to any picking-sequence for $n=2$. They prove that, when $n$ is unbounded, there can be exponentially many SPE, and finding even one of them is PSPACE-hard.

\citet{Aziz2017Equilibria} 
study the picking-sequence as a one-shot game, rather than as a sequential game. They assume that the agents report their entire ranking of the items to the principal, and the principal then picks, for each agent, his/her most preferred remaining item. 
They present a linear-time algorithm for computing a pure Nash equilibrium of this one-shot game.
They do not discuss whether the Nash equilibrium is also a competitive-equilibrium, or whether it satisfies any notion of fairness.
They assume that all agents have additive preferences, and all items are goods.

Another line of work related to picking-sequences asks how to select a picking-sequence that maximizes some global objective. 
\citet{Bouveret2011General} study this question  under the assumption that all items are goods, all preferences are additive, and moreover, there is a \emph{single} known common scoring-function that relates the rank of an item in an agent's ranking to its value for the agent. The allocator does not know the ranking of each agent, but knows that all rankings are random draws from a given probability distribution. The allocator's goal is to maximize the expected value of some social welfare function. They show picking-sequences that maximize the expected utilitarian welfare (sum of values) or the expected egalitarian welfare (minimum value) in various settings.
\citet{Kalinowski2013Social} show that, for two agents 
with a specific scoring-function called \emph{Borda scoring} \citep{Young1974Borda},
when each ranking is equally probable, the alternating sequence ($ABABAB...$) attains the maximum expected utilitarian welfare.

Algorithms \ref{alg:4 items 2 agents} and \ref{alg:4 items 3 agents} in this paper 
offer each agent in turn a possible pixep, 
and allow the agent to accept or reject that pixep.
If the agent accepts, the pixep is played; otherwise, the protocl offers another pixep to another agent.
This is similar but not identical to an \emph{alternating-offers} protocol for negotiation between two players. In an alternating-offers protocol,  
the players themselves make the offers: 
each player suggests a possible outcome until an agreement is reached \citep{anbarci1993noncooperative,anbarci2006finite,erlich2018negotiation}. 

\subsection{Fairness with Unequal entitlements}
As shown in Section \ref{sec:fairness}, competitive equilibrium with unequal incomes can be regarded as a rule for fair allocation of items among agents with unequal entitlements. 
Fair allocation with unequal entitlements has been studied in some recent papers.

\citet{farhadi2019fair} study allocation of indivisible goods in a cardinal model, in which the preferences of each agent $i$ are represented by an additive function $V_i$.
They define the \emph{weighted maximin share} of agent $i$. Their definition, translated to the notation of this paper, reads:
\begin{align*}
\text{WMMS}_i
~~
:= 
t_i\cdot
\left(
~~
\max_{\mathbf{Y}\in \partition{\allitems}{n}}
~~
\min_{j\in\allagents}
~~
\frac{V_i(Y_j)}{t_j}
\right)
\end{align*}
Intuitively, the partition $\mathbf{Y}$ is the fairest partition that agent $i$ could suggest if all agents had the same valuation function $V_i$.
An allocation $\mathbf{X}$ is called \emph{WMMS-fair} if $V_i(X_i)\geq \text{WMMS}_i$ for all $i\in\allagents$.

WMMS-fairness is fundamentally different than CE-fairness:
WMMS depends on the \emph{cardinal} values that the agents assign to the items, while the CE-fairness properties
\eqref{eq:1} and  \eqref{eq:2} depend only on the agents' \emph{ordinal} rankings of the bundles. 
The cardinal and ordinal fairness conditions are independent and do not imply each other.

To see that WMMS-fairness does not imply CE-fairness, consider an instance in which all entitlements are equal. Then
$\text{WMMS}_i = V_i\left(\maxmin{\allitems}{1}{n}\right)$,
so WMMS-fairness is equivalent to maximin-share fairness.
Suppose further that there are $m=n$ goods and all agents have the same valuations.
Then, the WMMS of all agents is the smallest value of a good, so every allocation that gives a single good to each agent is WMMS-fair. 
However, no allocation is envy-free, hence no allocation is CE-fair.

To see that CE-fairness does not imply WMMS-fairness, consider an instance with $n=2$ agents with entitlements $a=98$ and $b=2$.
Suppose there are $m=3$ goods $x,y,z$ that Alice values at $49, 49, 2$ and Bob values at $98,2,0$.
Hence,
$\text{WMMS}_A \geq 98\cdot 1 = 98$.
Similarly, Bob's WMMS partition is $(xz,y)$ so $\text{WMMS}_B \geq 2\cdot 1 = 2$.
It is impossible to give both agents their WMMS share, so a WMMS-fair allocation does not exist.
However, a CE does exist. For example, giving all goods to Alice and pricing them at $a/3$ is a CE, hence the allocation is CE-fair.

\citet{farhadi2019fair} 
present an algorithm that guarantees each agent $i$ a value of at least $\frac{1}{n}\text{WMMS}_i$, and prove it may be impossible to guarantee more.

\citet{aziz2019maxmin} extend this research to allocation of indivisible \emph{chores} among agents with additive preferences.
In this setting the values and the $\text{WMMS}_i$ are negative numbers.
They present an algorithm that guarantees to each agent $i$ a value of at least $n\cdot \text{WMMS}_i$, and 
prove that it may be impossible to guarantee more than $\frac{4}{3}\cdot \text{WMMS}_i$ even for two agents.
They show improved approximation ratios for some special cases.

Fair division with unequal entitlements 
has also been studied in \emph{cake-cutting}; see
\citet{cseh2018complexity,SegalHalevi2018Entitlements} and the references therein.
In this setting,
$\allitems$ is a continuous resource, usually represented by an interval and called ``cake''.
Each agent $i$ has a non-atomic measure  $V_i$ over the measurable subsets of $\allitems$.
Non-atomicity implies that, for each agent $i\in\allagents$ and each integer $d\geq 1$,  each bundle $Z$ can be partitioned into $d$ subsets of equal value to $i$. Hence, $V_i(\maxmin{Z}{l}{d}) = {l\over d}V_i(Z)$.
~~
There are two common fairness criteria for cake-cutting with different entitlements.

1. \emph{$\incomevector$-proportionality} means that for all $i\in\allagents$: $V_i(X_i) \geq \frac{t_i}{\sum_{j\in \allagents}t_j} V_i(\allitems)$.
It implies property \eqref{eq:1} for $K=\allagents$, but not for other $K$. For example, suppose there are three additive agents with equal entitlement. Suppose Alice values her bundle at $2$, Bob's bundle at $3$ and Carl's bundle at $1$. Then the allocation is $\incomevector$-proportional for her, but does not satisfy \eqref{eq:1} e.g. for $K=\{Bob\}$ and $l=d=1$.
The same example shows that $\incomevector$-proportionality does not imply \eqref{eq:2}, for example with $l_A=l_C=0,l_B=d_B=1$.

2. \emph{$\mathbf{t}$-envy-freeness} \citep{Reijnierse1998Finding}
means that for  all $i,j\in\allagents$: $V_i(X_i) \geq \frac{t_i}{t_j} V_i(X_j)$.
It implies $\incomevector$-proportionality, and implies \eqref{eq:1} for any $K\subseteq \allagents$ and any integers $d\geq l\geq 1$.
Assume that the income of agent $i$ satisfies $t_i\geq {l\over d}\sum_{j\in K}t_j$.
The 
$\incomevector$- envy-freeness implies that, for all $j\in \allagents$, $V_i(X_i)\cdot t_j \geq V_i(X_j)\cdot t_i$.
Summing over all $j\in K$ gives
$V_i(X_i)\cdot \sum_{j\in K}t_j \geq V_i(\cup_{j\in K}X_j)\cdot t_i \geq V_i(\cup_{j\in K}X_j)\cdot {l\over d}\sum_{j\in K}t_j$,
so 
$V_i(X_i) \geq {l\over d}V_i(\cup_{j\in K}X_j) = V_i(\maxmin{(\cup_{j\in K}X_j)}{l}{d})$.
Similarly, it implies \eqref{eq:2} for any sequence $d_j\geq l_j \geq 0$ ($j\in\allagents$).
Assume that the income of agent $i$ satisfies $t_i\geq \sum_{j\in K}{l_j\over d_j}t_j$.
The 
$\incomevector$-envy-freeness implies that, for all $j\in \allagents$, $V_i(X_i)\cdot {l_j\over d_j}t_j \geq {l_j\over d_j}V_i(X_j)\cdot t_i$.
Summing over all $j\in K$ gives
$V_i(X_i)\cdot \sum_{j\in K}{l_j\over d_j}t_j \geq \sum_{j\in K}{l_j\over d_j}V_i(X_j)\cdot t_i$,
so 
$V_i(X_i) \geq \sum_{j\in K}{l_j\over d_j}V_i(X_j) 
=
\sum_{j\in K}V_i(\maxmin{X_j}{l_j}{d_j})
=
V_i(
\sqcup_{j\in K}\maxmin{X_j}{l_j}{d_j})$.

The most recent advancement in this front was by \citet{garg2020approximating}. They presented an algorithm that approximates the maximum \emph{Nash welfare} for agents with different entitlements, defined as the weighted geometric mean of the agents' valuations:
\begin{align*}
\max_{\mathbf{X}}
\left(\prod_{i\in\allagents} V_i(X_i)^{t_i}\right)^{1/\sum_{i\in\allagents}t_i}
\end{align*}

\subsection{Apportionment methods}
\label{sub:apportionment}
The \emph{apportionment problem} 
is the problem of allocating seats in a legislature body among different groups in proportion to their size \citep{Balinski1975}.
Examples include allocating seats in the USA congress among the states in proportion to their population,
or allocating seats in a parliament among political parties in proportion to the number of ballots they received in the elections.
This is a special case of fair allocation with unequal entitlements,
where the items are the seats and the entitlements are the population sizes or the ballot counts.
In this special case, 
All items are identical,
all agents have identical additive preferences,
and the value of each bundle equals the bundle size.

A full treatment of apportionment methods is beyond the scope of the present paper; see \citet{Biro2015,Koczy2017} for  recent surveys. However, one method is particularly related to CE for almost all incomes, and it is the method developed by D'Hondt in 1882.%
\footnote{
I am grateful to an anonymous reviewer for suggesting the connection between the D'Hondt method and CE%
.
The method of D'Hondt has many variants that arrive at the same final outcome using a different algorithm. Some of them are the Jefferson method in the USA and the Bader-Ofer system in Israel.
See e.g. \citet{flis2019pot}.
}

Denote by $S$ the total number of seats to allocate.
For each party $i\in\allagents$, denote by $t_i$ its number of ballots and by $s_i$ the number of seats allocated to it.
Initially, all $s_i$ are set to $0$.
Then a seat is allocated to a party $i\in\allagents$ for which the quotient $t_i/(s_i+1)$ is largest, and its $s_i$ is increased by $1$. This repeats until finally $\sum_{i\in \allagents}s_i = S$.
Suppose that the vector $\incomevector$ of ballot-counts is such that
there are no identities of the kind:
\begin{align}
\label{eq:apportionment-ties}
k_i\cdot t_i = k_j\cdot t_j
\end{align}
for any two integers $k_i,k_j\range{1}{S}$ and any two indices $i\neq j$.%
\footnote{
I am grateful to user Wolfram here:
https://math.stackexchange.com/q/2118349/29780
for suggesting this restriction.
}
Then, at each step, there is a single $i$ for which 
the quotient $t_i/(s_i+1)$ is largest.
Let $j$ be the party who gets the last seat, and denote $p := t_j / s_j$.
Then the allocation $(s_1,\ldots,s_n)$ with the price $p$ for every seat is a CE. 
This is because, for all parties $i\in\allagents$:
\begin{align*}
&t_i / (s_i+1) ~<~ t_j/s_j ~\leq~ t_i/s_i
\\
\implies&
p\cdot s_i ~\leq~ t_i ~<~ p\cdot(s_i+1),
\end{align*}
so at a price of $p$ per seat, each party can afford its own bundle ($s_i$ seats) but cannot afford any larger bundle ($s_i+1$ seats).

The identities \eqref{eq:apportionment-ties} hold only for a subset of measure zero of $\allincomes$,
so if $\incomevector$ could be an arbitrary vector in $\allincomes$, then one could say that a CE exists for almost all incomes.
However, $\incomevector$ is in fact a vector of integers (ballot-counts), 
and all elements of $\incomevector$ are bounded by the total number of voters.
Hence the set of possible vectors $\incomevector$ is finite, 
so the 
identities \eqref{eq:apportionment-ties} hold for a positive fraction of this set.
However, ties such as \eqref{eq:apportionment-ties} are never encountered in real-life apportionment scenarios \citep{Koczy2018}, so practically, a CE always exists.

The above discussion implies that the D'Hondt seat allocation is CE-fair. 
In particular, applying 
\eqref{eq:1} with $K=\allagents$ and $d = S =$ the total number of seats
implies that,
if a party received at least a fraction $l/S$ of the total number of ballots,
then it must get at least the $l$-out-of-$S$ maximin-share of the $S$ seats, which is at least $l$ seats.
This property is called the \emph{lower Hare-quota}.
It is a known property of the D'Hondt's method  \citep{pukelsheim2017proportional}[Section 11.4].
But D'Hondt's method satisfies all the other instantiations of \eqref{eq:1} and \eqref{eq:2}, which --- as far as I know --- was not noted before.

It is interesting whether there are other apportionment rules that implement a CE-fair allocation,
and whether 
the existence result extends to settings in which different seats may have different values.

\subsection{Fairness in subgame-perfect equilibrium}
Most works on fair division either 
ignore strategic considerations altogether,
or look for \emph{truthful allocation mechanisms} --- mechanisms that implement a fair outcome in dominant strategies.
I am aware of only few works that consider implemetation of a fair outcome in subgame-perfect equilibrium:
\begin{itemize}
\item \citet{10.2307/1911364} 
show that, in general, various outcomes that cannot be implemented in Nash equilibrium, can be implemented in SPE.
One such outcome is the competitive equilibrium in a market of divisible resources (sub. 6.3).
\item \citet{Nicolo2008Strategic} present an algorithm for dividing a heterogeneous divisible resource (``cake'') between two agents, where the SPE is Pareto-efficient and envy-free;
\item \citet{Nicolo2012Equal} present an algorithm for the same setting, where the SPE satisfies a different fairness property known as \emph{equal-opportunity equivalence}; 
\item \citet{nicolo2017divide} present an algorithm for allocating a single indivisible item with monetary compensation between two agents, where the SPE
is  a unique allocation that would be obtained by a balanced market.
\end{itemize}
I am not aware of SPE implementations of fair allocation of indivisible items without money.

In this paper, the subgame-perfect equilibrium was mainly used as a tool to prove the existence of a CE, so it was sufficient to prove that \emph{there exists} a SPE which is also a CE. In economic terms, the algorithms in this paper provide a \emph{weak implementation} of CE.
A possible direction for future work is finding mechanisms 
in which \emph{every} SPE is a CE, i.e., mechanisms that provide a  \emph{strong implementation} of CE.

Besides subgame-perfect equilibrium, one could try to implement a CE using other solution concepts, such as 
Nash equilibrium (in a direct revelation game), 
self-confirming equilibrium \citep{fudenberg1993self}
or
rationalizable equilibrium \citep{bernheim1984rationalizable}.

\section{Future Work}
\label{sec:future}

\subsection{Restricted preference domains}
In the domains of monotonically-increasing and monotonically-decreasing preferences, 
we now know all the pairs $(n,m)$ for which a CE is guaranteed to exist for almost all incomes.
But for restricted preference domains, the question remains open. In particular,
for $n=2$ or $n=3$ agents with \emph{additive} preferences, when there are $m\geq 5$ goods, it is open whether CE exists for almost all incomes.
Note that the preferences in the negative result of Theorem
\ref{thm:2-agents-5-items} (for $n=2,m=5$) are not additive:
\begin{itemize}
\item For Alice, we had $xyz \succ vxy,vxz,vyz,wxy,wxz,wyz$; with additive preferences, this implies that each of $x,y,z$ is preferred to each of $v,w$, which implies $xy \succ vw$, which implies $xyz \succ vw$; but we had $vw \succ xyz$ --- a contradiction.
\item For Bob, we had $vx,vy,vz,wx,wy,wz \succ vw$; 
with additive preferences, this implies that each of $x,y,z$ is preferred to each of $v,w$.
But we had $vxy,vxz,vyz,wxy,wxz,wyz \succ xyz$, which implies the opposite.
\end{itemize}
\cite{Babaioff2017Competitive} provide several sufficient conditions for existence in the case of $2$ agents  with additive preferences, but a general solution is still not known.
Theorem 
\ref{thm:duality}
implies that, whenever a CE exists for two agents with additive preferences, the same is true for both goods and chores.

\subsection{Finding pixeps automatically}
\label{sub:auto-pixep}
The pixeps in this paper were found manually. It may be useful to build a program for automatically finding pixeps in domains in which the existence of CE is not settled yet (for example, two agents with additive preferences).
When the number of items is sufficiently small, it may be possible to check all picking-sequences, for each one of them calculate prices that satisfy the heuristics (H1,H2,H3) for some subset of the income-space,  and finally check whether the entire income-space is covered.%
\footnote{
A partial implementation of such a program, for the special case of two agents with identical additive preferences, can be found here:
https://github.com/erelsgl/indivisible-competitive-equilibrium
}

Unfortunately, while such a program may help to prove the existence of CE, it cannot be used to prove non-existence. As shown in the following example, some competitive equilibria cannot be a SPE of a pixep satisfying the decreasing-prices heuristic.
\begin{example}
There are 4 goods and 2 agents with incomes $(a,b)=(8,7)$. The agents' preferences contain the following relations:
\begin{itemize}
\item Alice: $wx \succ wy \succ wz \succ xy \succ xz \succ yz$;
\item Bob:   $xy \succ xz \succ yz \succ wx \succ wy \succ wz$.
\end{itemize}
Consider the allocation $(A,B) = (wz,xy)$. It is a CE, for example, with price-vector $(p_w,p_x,p_y,p_z) = (6,4,3,2)$.
Note that $p_x > p_z$ and $p_y > p_z$. This is necessary for a CE, since otherwise Alice could afford $wx$ or $wy$, which she prefers to $wz$.
Note also that $p_w > p_x$ and $p_w > p_y$. This is also necessary for a CE in which agents exhaust their incomes, since $p_w = a - p_z > b - p_y = p_x$ and similarly $p_w > p_y$. 

If such a CE could be a SPE of a pixep with decreasing-prices, then the sequence should be ABBA.
But the allocation $(wz,xy)$ is not a subgame-perfect equilibrium of this sequence: in SPE, Alice picks $x$ first, since then Bob picks $yz$ and Alice gets $wx$, which she prefers to $wz$.
\qed
\end{example}
Thus, the challenge of automatically solving the question of CE existence for almost incomes remains open.

\begin{acknowledgements}
I am grateful to 
Inbal Talgam-Cohen, Noam Nisan, Moshe Babaioff, Fedor Sandomirskiy, lonza leggiera, 
Wolfram, and
three anonymous reviewers of AAMAS 2018 for their very helpful comments. Special thanks are due to the anonymous referees of JAAMAS for their extraordinarily helpful reviews.
\end{acknowledgements}

\newpage

\appendix
\section{Relations Between Different CE Fairness Properties}
\label{app:fairness}
The CE fairness properties \eqref{eq:1} 
and \eqref{eq:2} imply 
many different fairness conditions that are satisfied by a CE allocation. This appendix analyzes the relations between different such conditions: which of them are independent and which of them are implied by others.

As in Section \ref{sec:fairness}, there is a fixed allocation $\mathbf{X}$ and a specific agent, Alice, with income $a$, preference-relation $\succeq$, and bundle $A$.

\subsection{Maximin share with a bounded number of items}
\label{sub:indep-items}
The following lemma shows that, 
the number of ``interesting'' instantiations of \eqref{eq:1} and \eqref{eq:2} is finite, even though the number of fractions $l/d$ satisfying their left-hand side is infinite. 

\begin{lemma}
\label{lem:mms-bundlesize}
For every integers $h\geq 0$ and $1\leq l\leq d$, If $|X| \leq  d$ then:
$
\maxmin{X}{l}{d}
\succeq
\maxmin{X}{l+h}{d+h}.
$
\end{lemma}
\begin{proof}
Let $W' := \maxmin{X}{l+h}{d+h}$.
By definition of the maximin share, 
there exists a partition $\mathbf{Y'}\in \partition{X}{d+h}$,
such that $W' = \min_\succ \union{\mathbf{Y'}}{l+h}$.

Since $|X|\leq d$, at most $d$ parts in $\mathbf{Y'}$ are non-empty.
Consider a new partition $\mathbf{Y}\in\partition{X}{d}$,
which contains all non-empty parts in $\mathbf{Y'}$.
Then, each bundle $W\in \union{\mathbf{Y}}{l}$
is also contained in $\union{\mathbf{Y'}}{l+h}$, so
$W\succeq W'$. 
This is true, in particular, when $W$ is the minimum bundle in $\union{\mathbf{Y}}{l}$, which is by definition $\maxmin{X}{l}{d}$.
\qed
\end{proof}
This means that, to check whether an allocation is CE-fair, it is sufficient to check instantiations of \eqref{eq:1} and 
\eqref{eq:2} in which the denominator of each fraction is at most the number of items in the corresponding bundle. For example, to check \eqref{eq:1} for Alice w.r.t Bob, if $|B|=5$ then it is sufficient to check fractions $l/d$ with $a \geq {l\over d} b$ and $d\in\{1,2,3,4,5\}$, since when $d>5$, the condition $A\succeq \maxmin{B}{l}{d}$ is implied e.g. by the condition $A\succeq \maxmin{B}{l-d+5}{5}$.
In the 
technical report\footnote{{http://arxiv.org/abs/1912.08763}} it is further shown that only a small number of these fractions need to be checked.

\subsection{Independence of \eqref{eq:1} instantiations with different bundles}
\label{sub:indep-p1-bundles}
Consider the instantiations of \eqref{eq:1}
with different bundles on the right-hand side.
Fix a subset of agents $K\subseteq\allagents$ that does not contain Alice, and two integers $l,d$. Then the following two conditions on Alice's income are equivalent:
\begin{align*}
a \geq {l\over d} \sum_{i\in K}t_i 
&& \iff && 
a \geq {l\over l+d} (a+\sum_{i\in K}t_i )
\end{align*}
By \eqref{eq:1}, each of these conditions induces a different fairness condition on Alice's bundle:
\begin{align*}
A \succeq {\maxmin{\left(\bigcup_{i\in K} X_i\right)}{l}{d}}
&& \text{vs.} && 
A \succeq {\maxmin{\left(A\cup \bigcup_{i\in K} X_i\right)}{l}{l+d}}
\end{align*}
The following lemma shows that, in general, these two fairness conditions are  independent --- none of them implies the other.
Below, $X_K := \left(\bigcup_{i\in K} X_i\right)$.
\begin{lemma}
For every subset $K\subseteq \allagents$ that does not contain Alice, and for every two integers $l,d$ with $1\leq l\leq d/2$:

(i) There exists an allocation $\mathbf{X}$ and a preference-relation $\succ$ by which
\begin{align*}
{\maxmin{\left(A\cup X_K\right)}{l}{l+d}} 
~\succ~ A ~\succeq~
{\maxmin{\left(X_K\right)}{l}{d}}
\end{align*}

(ii) There exists an allocation $\mathbf{X}$ and a preference-relation $\succ$ by which 
\begin{align*}
{\maxmin{\left(X_K\right)}{l}{d}} 
~\succ~ A ~\succeq~
{\maxmin{\left(A\cup X_K\right)}{l}{l+d}}
\end{align*}
\end{lemma}
\begin{proof}
Both parts are proved using additive preference-relations.
For simplicity, the examples have different items with the same value; it is easy to make the preferences strict by adding a very small ``noise'' $\epsilon_j>0$ to each item $j$. This does not affect the proof arguments.

\emph{(i)} 
Consider any allocation in which Alice's bundle $A$ contains $l+d$ items that she values at $2 l$,
and $X_K$ contains $l+d$ items that Alice values at $2 d + 1$.

The union $A\cup X_K$ 
can be partitioned into $l+d$ pairs, each of which is valued by Alice at $2l+2d+1$. Hence, Alice values the bundle ${\maxmin{\left(A\cup X_K\right)}{l}{l+d}}$ 
at $l\cdot (2l+2d+1) = 2l\cdot d + 2l\cdot l+ l$.

Alice values her own bundle at $2l\cdot (l+d) = 2l\cdot d + 2l\cdot l$.

Since $l\leq d/2$, any partition of $X_K$ into $d$ parts contains at least $l$ parts with at most a single item. Therefore, Alice values the bundle ${\maxmin{\left(X_K\right)}{l}{d}}$ at $l\cdot(2d+1) = 2l\cdot d+l$.

Since $2l\cdot d + 2l\cdot l+ l ~>~ 2l\cdot d + 2l\cdot l ~\geq~ 2l\cdot d+l$, the claim is proved.

\emph{(ii)} 
Consider any allocation in which Alice's bundle $A$ contains a single item that she values at $1$,
and $X_K$ contains $d$ items that Alice values at $2$.

Alice values the bundle ${\maxmin{\left(X_K\right)}{l}{d}}$ at $2 l$. She values her own bundle at $1$.
The union $A\cup X_K$ contains $d+1$ items, so in any partition of it into $l+d$ parts, there are at least $l-1$ empty parts.
Hence, the $l$ lowest-valued parts contain at most $1$ non-empty part, and if such part exists, it must be Alice's item which she values at $1$. Hence, Alice values the bundle ${\maxmin{\left(A\cup X_K\right)}{l}{l+d}}$ at $1$.
Since $2l > 1 \geq 1$, the claim is proved. \qed
\end{proof}

\subsection{Independence of \eqref{eq:1} and \eqref{eq:2}}
\label{sub:indep-p1-p2}
Consider the possible instantiations of \eqref{eq:2}. Fix a subset of agents $K\subseteq\allagents$ that does not contain Alice, and two integers $l,d$ with $1\leq l\leq d$. Then the following two conditions on Alice's income are equivalent:
\begin{align*}
a \geq {d\over d+l} \left(a+\sum_{i\in K}t_i \right)
&& \iff && 
a \geq {d-l\over d} a + \sum_{i\in K}t_i 
\end{align*}
The following fairness conditions are implied by \eqref{eq:1} and  \eqref{eq:2} respectively (where again $X_K := \left(\bigcup_{i\in K} X_i\right)$):
\begin{align*}
A \succeq {\maxmin{\left(A\cup X_K\right)}{d}{d+l}}
&& \text{vs.} && 
A \succeq 
\left(\maxmin{A}{d-l}{d}\right)\sqcup X_K
\end{align*}

\begin{lemma}
\label{lem:p2-independent}
For every subset $K\subseteq \allagents$ that does not contain Alice, and for every two positive integers $l,d$ with $d/2\leq l < d$:

(i) There exists an allocation $\mathbf{X}$ and a preference-relation $\succ$ by which
\begin{align*}
\left(\maxmin{A}{d-l}{d}\right)\sqcup X_K
\succ A \succeq 
{\maxmin{\left(A\cup X_K\right)}{d}{d+l}}
\end{align*}

(ii) There exists an allocation $\mathbf{X}$ and a preference-relation $\succ$ by which 
\begin{align*}
{\maxmin{\left(A\cup X_K\right)}{d}{d+l}}
\succ A \succeq 
\left(\maxmin{A}{d-l}{d}\right)\sqcup X_K
\end{align*}
\end{lemma}
\begin{proof}

\emph{(i)} 
Consider an allocation in which Alice's bundle $A$ contains $d$ items that she values at $1$,
and $X_K$ contains one item that Alice values at $l+1$.

The bundle $\left(\maxmin{A}{d-l}{d}\right)\sqcup X_K$ contains arbitrary $d-l$ items from $A$  and the bundle $X_K$, so its value for Alice is $d-l+l+1 = d+1$.
Alice values her own bundle at $d$.
The union $A\cup X_K$ contains only $d+1$ items.
Hence, in the optimal partition of it into $d+l$ parts,
$l-1$ parts are empty and each of the other $d+1$ parts contains exactly one item. 
The bundle ${\maxmin{\left(A\cup X_K\right)}{d}{d+l}}$ contains the $d$ worst items, which Alice values at $1$, so Alice values that bundle at $d$.
Since $d+1 > d \geq d$, the claim is proved.

\emph{(ii)} 
Consider an allocation in which Alice's bundle $A$ contains 
$d+l$ items that she values at $d+l$,
and $X_K$ contains $d+l$ items that Alice values at $2l$.

The union $A\cup X_K$ 
can be partitioned into $d+l$ pairs, each of which is valued by Alice at $d+l+2l$. Hence, Alice values the bundle ${\maxmin{\left(A\cup X_K\right)}{d}{d+l}}$ 
at $d\cdot (d+3l) = d\cdot(d+l) + 2d\cdot l$.

Alice values her own bundle at $(d+l)\cdot (d+l) = d\cdot(d+l) + (d+l)\cdot l$.

The condition $d/2\leq l$ implies $d-l\leq d/2$. 
So any partition of $A$ into $d$ parts 
contains at  least $d-l$ parts with at most a single item. Therefore, Alice values the bundle $\left(\maxmin{A}{d-l}{d}\right)\sqcup X_K$ 
at 
$(d-l)\cdot(d+l) + (d+l)\cdot 2 l = (d+l)\cdot (d+l) = d\cdot(d+l) + (d+l)\cdot l$.

Since $d>l$, we have $d\cdot(d+l) + 2d\cdot l > d\cdot(d+l) + (d+l)\cdot l$ and the claim is proved.
 \qed
\end{proof}

Using the construction of Lemma \ref{lem:p2-independent}, it is possible to show that \eqref{eq:1} and \eqref{eq:2} are independent --- none of them implies the other one even for two additive agents.

\begin{lemma}
\label{lem:indep-p1-p2}
(i) There is an instance with $n=2$ additive agents such that, for one agent, some allocation satisfies all instantiations of \eqref{eq:1} but violates an instantiation of \eqref{eq:2};

(ii) There is an instance with $n=2$ additive agents such that, for one agent, 
some allocation satisfies all instantiations of \eqref{eq:2} but violates an instantiation of \eqref{eq:1}.
\end{lemma}
\begin{proof}
There are two agents with incomes $a,b$ satisfying:
\begin{align*}
2b < a < 3b
\end{align*}
(e.g. $a=9,b=4$).
Each part of the lemma is proved by a different instance.

\emph{(i)} There are 3 goods: $x,y,z$, and Alice's preferences are:
\begin{align*}
xyz > xy > xz > yz > x > y > z > \emptyset
\end{align*}
(they are additive, for example where $x,y,z$ are valued at $20,15,10$. Bob's preferences are not relevant for the proof).

The allocation is:
\begin{align*}
A = \{y,z\} &&  B = \{x\}
\end{align*}

For \eqref{eq:1}, the relevant bundles in the right-hand side are $B$ and $A\cup B$.
For $B$, 
since $a > b$, \eqref{eq:1} implies that $A \succeq B$, which is satisfied. 
The bundle $A\cup B$ contains three items,  so
 by Lemma \ref{lem:mms-bundlesize} only conditions with $d\in\{1,2,3\}$ are interesting. 
The strongest conditions implied by the incomes for this bundle are $A\succeq \maxmin{(A\cup B)}{1}{2} = \{x\}$ and $A\succeq \maxmin{(A\cup B)}{2}{3} = \{y,z\}$, both of which are satisfied.
However, \eqref{eq:2} is violated, since:
\begin{align*}
a > {1\over 2}a + b
&&\text{while}&&
A\prec \left(\maxmin{A}{1}{2}\right)\sqcup B = \{x,z\}
\end{align*}

\emph{(ii)} There are 6 goods,
$u,v,w,x,y,z$. Alice's preferences are additive and she values the goods at $40,35,30,25,20,15$ respectively. 
The allocation is:
\begin{align*}
A = \{u,v,w\} && B = \{x,y,z\}
\end{align*}
so Alice values $A$ at 105 and $B$ at 60.
For \eqref{eq:2}, the bundles at the right-hand side are $A$ and $B$.
Since each bundle contains three identical items,
by Lemma \ref{lem:mms-bundlesize} only conditions with $d\in\{1,2,3\}$ are interesting, so the relevant MMS bundles  are $\maxmin{A}{1}{3}$,  $\maxmin{A}{1}{2}$, $\maxmin{A}{2}{3}$, $A$, and
$\maxmin{B}{1}{3}$, $\maxmin{B}{1}{2}$, $\maxmin{B}{2}{3}$, $B$. By checking all 16 combinations of these bundles, 
it can be seen that the strongest conditions implied by the incomes (since $a < {2\over 3}a + b$) are:
\begin{align*}
A \succeq 
\left(\maxmin{A}{1}{2}\right) \sqcup B = 
\{u,x,y,z\}
&&
A \succeq 
\left(\maxmin{A}{2}{3}\right) \sqcup \left(\maxmin{B}{2}{3}\right)
=
\{v,w,y,z\}
\end{align*}
and both of them are satisfied (Alice values both bundles at 100).
However, \eqref{eq:1} is violated, since $a > 2b$ implies:
\begin{align*}
a > {2\over 3}(a + b)
&&\text{while}&&
A\prec \maxmin{(A\cup B)}{2}{3} = \{u,v,y,z\}
\end{align*}
(Alice values that bundle at 110).
\qed
\end{proof}

\section{CE-fairness and Pareto-efficiency do not imply CE}
\label{sec:ce-fair-not-ce}
Since properties \eqref{eq:1} and \eqref{eq:2} imply so many different fairness conditions, one could think that the combination of all these fairness conditions, together with Pareto-efficientity, is equivalent to CE. However, the following lemma shows that this is not true.
\begin{lemma}
There exists an instance with $n=2$ agents and $m=4$ goods, a subset of the income-space with a positive measure, and an allocation $\mathbf{X}$, 
such that:

(i) $\mathbf{X}$ is Pareto-efficient;

(ii) $\mathbf{X}$ satisfies \eqref{eq:1} for both agents, for every $K\subseteq \allagents$ and all integers $l,d$;

(iii) $\mathbf{X}$ satisfies  \eqref{eq:2} for both agents, for every combinations of integers $l_1,l_2,d_1,d_2$.

(iv) $\mathbf{X}$ is not a CE allocation for any price-vector.
\end{lemma}

\begin{proof}
There are two agents with incomes $a,b$ satisfying:
\begin{align*}
2a/3 < b < a
\end{align*}
which is equivalent to:
\begin{align*}
2(a+b)/5 ~<~ b ~<~ (a+b)/2 ~<~ a ~<~ 3(a+b)/5
\end{align*}
(e.g. $a=9,b=7$).
The agents have identical preferences containing the following relations:
\begin{align*}
wx > wy > wz > xy > xz > w > yz > x > y > z
\end{align*}
(they are additive, for example where $w,x,y,z$ are valued at $8,6,4,3$).
The allocation is:
\begin{align*}
A = \{w,x\} &&  B = \{y,z\}
\end{align*}

\emph{(i)} Since the preferences are identical, any allocation is Pareto-efficient. 

\emph{(ii)} 
For Alice, the relevant bundles in the right-hand side of \eqref{eq:1} are $B$ and $A\cup B$.
For $B$, \eqref{eq:1} is satisfied since $A \succeq B$.
The bundle $A\cup B$ contains four items, 
so by Lemma \ref{lem:mms-bundlesize} only conditions with $d\in\{1,2,3,4\}$ are interesting. 
Since $a < {1\over 1}\cdot(a+b)$ and $a < {2\over 2}\cdot(a+b)$ and $a < {2\over 3}\cdot(a+b)$ and $a < {3\over 4}\cdot(a+b)$, the relevant instantiations of \eqref{eq:1}  are:
\begin{align*}
A \succeq \maxmin{\left(A\cup B\right)}{1}{2} = \{x,y\}
&&
A \succeq \maxmin{\left(A\cup B\right)}{1}{3} = \{x\}
&&
A \succeq \maxmin{\left(A\cup B\right)}{2}{4} = \{y,z\}
\end{align*}
all of which are satisfied.

For Bob, the relevant bundles at the right-hand side of \eqref{eq:1} are $A$ and $A\cup B$.
Since $A$ contains 2 items, by Lemma \ref{lem:mms-bundlesize} only conditions with $d\in\{1,2\}$ are interesting. 
Since $b<{1\over 1}a$ and $b<{2\over 2}a$,
the only relevant con condition is $B \succeq \maxmin{A}{1}{2} = \{x\}$, which is satisfied.
For the bundle $A\cup B$, 
since $b < {1\over 2}\cdot(a+b)$ and $b < {2\over 3}\cdot(a+b)$ and $b < {2\over 4}\cdot(a+b)$, the relevant instantiations of \eqref{eq:1}  are:
\begin{align*}
B \succeq \maxmin{\left(A\cup B\right)}{1}{3} = \{x\}
&&
B \succeq \maxmin{\left(A\cup B\right)}{1}{4} = \{z\}
\end{align*}
both of which are satisfied.

\emph{(iii)}
The bundles in the right-hand side of \eqref{eq:2} are $A$ and $B$, and each of them has 2 items. By Lemma \ref{lem:mms-bundlesize}, the only interesting values of $d_1,d_2$ are $1$ and $2$.

For Alice, 
since $a < {1\over 2}a + {1\over 1}b$,
the only relevant condition is
$
a > {1\over 2}a + {1\over 2}b 
$
which implies:
\begin{align*}
A \succeq \left(\maxmin{A}{1}{2}\right)\sqcup\left(\maxmin{B}{1}{2}\right) = \{x,z\}
\end{align*}
which is satisfied.
For Bob, since $
b < {1\over 2}a + {1\over 2}b$, \eqref{eq:2} does not lead to any non-trivial condition.

\emph{(iv)}
Any CE price $\mathbf{p}$ must satisfy 
$p(w) > b$, since Bob prefers $w$ to his bundle. 
since $p(A)\leq a$, it must satisfy $p(x) < a-b$. 
On the other hand, since $p(B)\leq b$, 
either $p(y)$ or $p(z)$ is at most $b/2$.
Hence, either $p(x y)$ or $p(x z)$ is at most $a-b+b/2 = a-b/2 < (3b/2)-b/2 = b$.
So Bob can afford either $x y$ or $x z$, both of which he prefers to his current bundle.
%
%
\footnote{
A property that captures this issue is: ``if $b > (a-b) + b/2$, then for every partition of $A$ into two parts $A_1\cup A_2$,
Bob prefers $B$ either to $A_1$, or to $A_2 \cup \left(\maxmin{B}{1}{2}\right)$''. This condition is much less natural than properties \eqref{eq:1} and \eqref{eq:2}, and it is not clear how to generalize it.
}
\qed
\end{proof}

\section{Existence of CE: Alternative Definitions}
\label{sec:alternative-defs}
The existence results in most of the paper were based on Definition \ref{def:ce-exists-in-n-m}, which requires that:

\begin{quote}
(*)
For \emph{almost all} income-vectors, for \emph{all} preference-profiles, a CE exists. 
\end{quote}
By changing the quantifiers, one could come up with several alternative definitions:
\begin{quote}
(**)
For \emph{all} income-vectors, for \emph{almost all} preference-profiles, a CE exists. 

(***) 
For \emph{all} preference-profiles, for \emph{almost all} income-vectors, a CE exists. 

(****) 
For \emph{almost all} preference-profiles, for \emph{all} income-vectors, a CE exists. 
\end{quote}

Definitions (**) and (****) are both equivalent to ``for \emph{all} income-vectors, for \emph{all} preference\-/profiles, a CE exists''. This is because the space of preference-profiles is discrete and finite, so even a single preference-profile constitutes a positive fraction of this space. Therefore they cannot be guaranteed even for two agents and one item.

Definition (*) clearly implies (***).%
\footnote{
This is a general fact.
For every predicate $P(x,y)$, 
the statement ``For-almost-all $x\in X$ for-all $y\in Y$ $P(x,y)$''
implies ``for-all $y\in Y$ for-almost-all $x\in X$ $P(x,y)$''.
}
At first glance, definition (***) is weaker than (*), since it allows the set of excluded incomes (for which CE is allowed to not exist) to depend on the preference-profile.
For example, consider the following hypothetic claim: ``If Alice prefers $x$ to $y$, then CE exists whenever $a \neq b$; if Alice prefers $y$ to $x$, then CE exists whenever $a\neq 2 b$''. Formally, this claim fits definition (***) but not definition (*).

However, because the number of preference-profiles is finite, the two definitions are in fact equivalent: for each preference-profile, the set of excluded incomes has zero measure, and the union of a finite number of such sets still has zero measure. For example, the above claim implies that ``CE exists whenever $a \neq b$ and $a\neq 2 b$, for all preference-profiles'', which formally fits definition (*).%
\footnote{
This is also a general fact.
If the set $Y$ is finite, then 
the statement 
``For-almost-all $x\in X$ for-all $y\in Y$ $P(x,y)$''
is equivalent to
``for-all $y\in Y$ for-almost-all $x\in X$ $P(x,y)$''.
}

\section{Efficiency of CE with strict and non-strict preferences}
\label{sec:indifferences}
An allocation is called 
\emph{Pareto-efficient (PE)} if no other allocation is weakly-preferred by all agents and strictly-preferred by at least one agent.
Let $(\pricevector,\mathbf{X})$ be a CE. 

With strict preferences, $\mathbf{X}$ is PE. 
\emph{Proof}:
Suppose for contradiction that a different allocation $\mathbf{Y}$ is strictly preferred by some subset of agents $J\subseteq\allagents$.
CE condition 2 on the CE $(\pricevector,\mathbf{X})$ implies that $\forall j\in J: p(Y_j)>p(X_j)$.
Moreover, all other agents are indifferent between 
$\mathbf{Y}$ and $\mathbf{X}$. Since the preferences are strict, this implies that $\forall i\in\allagents\setminus J: Y_i = X_i$, so $p(Y_i)=p(X_i)$.
Hence $p(\cup_{i\in\allagents}Y_i) > p(\cup_{i\in\allagents}X_i)$.
But this is impossible as 
$\cup_{i\in\allagents}Y_i = \cup_{i\in\allagents}X_i = \allitems$.

With indifferences, 
$\mathbf{X}$ may not be PE.
For example, suppose there are two goods $x,y$, and two agents: Alice with income $2/3$ and Bob with income $1/3$.   Alice is indifferent between $x$ and $y$, while Bob strictly prefers $y$.   Consider the allocation in which $y$ is priced at $2/3$ and given to Alice, and $x$ is priced at $1/3$ and given to Bob. This is a CE, but it is not Pareto-efficient because of the allocation in which $x$ is given to Alice and $y$ is given to Bob.

An allocation is called \emph{weakly-Pareto-efficient (WPE)} if no other allocation is strictly preferred by all agents. 
Every CE allocation is WPE even with  indifferences.
\emph{Proof}:
Suppose for contradiction that a different allocation $\mathbf{Y}$ is strictly preferred by all agents. 
CE condition 2 on the CE $(\pricevector,\mathbf{X})$ implies that $p(Y_i)>p(X_i)$ for all $i\in\allagents$.
Hence $p(\cup_{i\in\allagents}Y_i) > p(\cup_{i\in\allagents}X_i)$.
But this is impossible as 
$\cup_{i\in\allagents}Y_i = \cup_{i\in\allagents}X_i = \allitems$.

\section{Is a CE scarcer when there are more agents?}
\label{sec:more-agents}

Intuitively, one could think that, in  instances with more agents, a CE is ``scarcer'' ---  less likely to exist. But this intuition is not always true. 

\begin{example}
Consider an instance with three goods and three agents: Alice Bob and Carl, with incomes $a = b > c > 0$. The preference-profile contains the following relations:

\begin{itemize}
\item     Alice:  $xy \succ yz \succ xz \succ x \succ y \succ z\succ \emptyset$.
\item     Bob:~  $xy \succ xz \succ yz \succ y \succ x \succ z\succ \emptyset$.
\end{itemize}
The allocation Alice:$x$, Bob:$y$, Carl:$z$ with price-vector $(a,b,c)$ is a CE. 

However, if Carl leaves then there is no CE, since either Alice or Bob gets a single item and envies the other agent who gets two items.

Moreover, even if Carl, upon leaving, gives his income and his item to one of the other agents, the result is still not a CE. For example, the allocation Alice:$xz$, Bob:$y$ where Alice's income is $a+c$ and Bob's income is $b$ is not a CE, since with this larger income, Alice can afford $yz$, which she prefers to $xz$.
\end{example}

The above example is a knife-edge case since $a=b$.
It could still be the case that, for \emph{almost all} income-vectors,
CE is scarcer in instances with more agents. Formally, one could make the following conjectures (analogous to Lemma \ref{lem:more-items}):

\begin{conjecture}
\label{cnj:more-agents}
Let $n\geq 2$ and $m\geq 2$ be integers.

(a)
If a CE exists with $n+1$ agents, for almost-all
income-vectors in $\allincomesplus$
and all preference-profiles in $\allpreferencesplusn$,
then a CE exists with $n$ agents, for almost-all
income-vectors in $\allincomes$
and all preference-profiles in $\allpreferences$.

(b)
If a CE-fair allocation exists with $n+1$ agents, for almost-all
income-vectors in $\allincomesplus$
and all preference-profiles in $\allpreferencesplusn$,
then a CE-fair allocation exists with $n$ agents, for almost-all
income-vectors in $\allincomes$
and all preference-profiles in $\allpreferences$.
\end{conjecture}

So far I could not prove any of these conjectures.

\section{Proof Comparison}
\label{sec:proof-comparison}
This section compares two proofs to the corresponding proofs of \citet{Babaioff2017Competitive}.

\subsection{Theorem \ref{thm:3-items} --- existence of CE with 3 goods}
\citet{babaioff2019competitive}~[Proposition 3.1] 
prove this theorem by partitioning the income-space into five sub-spaces (instead of two). 
While they do not use pixeps, their algorithm for finding a CE in these sub-spaces can be, approximately, presented by the following pixeps:
\begin{enumerate}
\item If ~~ ${a > 3 b}$ ~~ then ~~ 
${\pick{A}{a/3}~~~\pick{A}{a/3}~~~~~~~~~\pick{A}{a/3}}$~~~~.
\\[3mm]
\item If ~~ ${3 b \geq a > 2 b}$ ~~ then ~~ 
${\pick{A}{a/2}~~~\pick{A}{a/2}~~~~~~~~~\pick{B}{b}}$~~~~.
\\[3mm]
\item If ~~ ${2b \geq a > b+c}$ ~~ then ~~ 
play the sequential game below:
\begin{align*}
&\text{Bob may choose:} && {\pick{A}{a/2}~~~\pick{A}{a/2}~~~~~~~~~\pick{B}{b}}~~~~.
\\[3mm]
&\text{Else:} && 
{\pick{A}{b^+}~~~\pick{B}{b}~~~~~~~~~\pick{A}{(a-b)^-}}~~~~.
\end{align*}
\item If ~~ ${a = b+c}$ ~~ then ~~ 
play the sequential game below:
\begin{align*}
&\text{Alice may choose:} && {\pick{A}{a}~~~\pick{B}{b}~~~~~~~~~\pick{C}{c}}~~~~.
\\[3mm]
&\text{Else:} && 
{\pick{B}{b}~~~\pick{A}{a/2}~~~~~~~~~\pick{A}{a/2}}~~~~.
\end{align*}
\item If ~~ ${b + c > a}$ ~~ then ~~ 
${\pick{A}{a}~~~\pick{B}{b}~~~~~~~~~\pick{C}{c}}$~~~~.
\end{enumerate}
~\\
The income-subspaces of steps 1, 2, 3, 5 are covered by the two steps of Algorithm \ref{alg:3 items}, but step 4 handles the set of incomes with $a=b+c$, which is not covered by Algorithm \ref{alg:3 items}. I omitted it since it has a measure of zero, so it is not required for proving the existence of CE for almost all incomes.

Here is a proof that step 4 indeed implements a CE. 
First, note that the decreasing-prices condition is satisfied since $a>b>c$ and $2b > b+c = a$.
By Lemma \ref{lem:dominated} no agent wants a dominated bundle.
Bob has an unrelated bundle only in the second pixep, and he cannot afford it.
It remains to check the unrelated bundles of Alice.
Rename the items such that Alice's best item is $x$.
If Bob's best item is not $x$, then Alice certainly plays $BAA$, since then she gets $x$ plus another item. She obviously prefers this bundle to her only unrelated bundle, which is Bob's single item.
If Bob's best item is $x$, then Alice effectively chooses between $x$ (the $ABC$ sequence) and $yz$ (the $BAA$ sequence). Whatever she chooses, she does not want the other (unrelated) bundle.

\subsection{Theorem \ref{thm:2-agents-5-items} --- non-existence of CE with 5 goods}
\citet{babaioff2019competitive}~[Theorem 3.5] 
prove this theorem using a very similar example. With the transformation $A\to v, B\to w, C\to x, D\to y, E\to z$, the preferences implied by the numeric values in their Example 3.1 are:
\begin{itemize}
\item Alice:    $vwz > vwy > vwx > vw > xyz > zyw > zxw > yxw > zyv > zxv > yxv > zy > zx > yx > zw > yw > xw > zv > yv > xv > z > y > x > w > v$.
\item Bob:   $wzy > wzx > wyz > vzy > vzx > vyx > zwv > ywv = zw > xwv = yw = zv > xw = yv > xv > xyz > vw > w > v > zy > zx > yx > z > y > x$.
\end{itemize}
The preferences are very similar, but they contain some indifferences between bundles (which do not affect the proof in any way).   I just expressed the preferences as relations and without numbers, and also did not specify some relations that are irrelevant for the proof.

The main differences between the proofs are:
(a) the proof in this paper uses only the fairness properties, so the impossibility result is stronger --- it shows that even the fairness properties alone cannot be satisfied;
(b) the proof in this paper is extended to any number $n \geq 2$ of agents.

\section{Allocations in Different Subgame-perfect Equilibria}
\label{sec:spe}
As noted in Section \ref{sec:3items}, some pixeps may have several different subgame-perfect equilibria, 
where only one of these SPE leads to a competitive equilibrium.  
Currently I do not have a general algorithm for determining the SPE that leads to a CE. 
However, the following lemma shows that, with strict preferences, the SPE selection affects only the prices and not the allocation.
\begin{lemma}
\label{lem:spe}
Suppose all agents have strict preferences.
Then for any two subgame-perfect equilibria $Q_1,Q_2$ of the same pixep $S$, the allocations $\mathbf{X}(S,Q_1)$ and $\mathbf{X}(S,Q_2)$ are the same.
\end{lemma}
\begin{proof}
The proof is by induction on $n$. For $n=1$ the claim is trivial. For $n>1$, assume the claim for all pixeps with $n-1$ agents, and consider two SPE of a pixep with $n$ agents. Suppose w.l.o.g. that 
the first choice that differs between $Q_1$ and $Q_2$ is made by Alice. If Alice's bundle in $Q_1$ differs from her bundle in $Q_2$, then (by strictness of preferences) she prefers one of these bundles; but then the other one cannot be a SPE.
So Alice must get the same bundle in $Q_1$ and $Q_2$. Therefore,  we can remove Alice from the sequence and remove her bundle from the set of items and get a new pixep with $n-1$ agents; by the induction assumption, all these agents have the same  bundle in both equilibria. \qed
\end{proof}
Lemma \ref{lem:spe} implies, in particular, that the resulting SPE allocation is CE-fair, since CE-fairness depends only on the allocation; the price-vector is used only as an evidence that the allocation is a CE.

In general, Lemma \ref{lem:spe} 
does not hold for a sequential choice among several pixeps.
For example, in the game ``Alice chooses between $AB$ and $AC$'', there are two SPEs with substantially different allocations.
However, the lemma does hold in the specific sequential games in Algorithms \ref{alg:4 items 2 agents} and \ref{alg:4 items 3 agents}:
\begin{itemize}
\item In Algorithm \ref{alg:4 items 2 agents} there are only two agents: if Alice's bundle is the same in both SPEs then obviously Bob's bundle is the same too.
The same is true in Algorithm \ref{alg:4 items 3 agents} Range 4.
\item In Algorithm \ref{alg:4 items 3 agents} Range 5, Alice cannot be indifferent between her two choices, since in each choice she gets a different bundle (a singleton vs. a pair), and by assumption the preferences are strict.
The same is true in Range 7.
In Range 6, too, Alice's choice is between a singleton and a pair, and Bob's choice is between a pair and a singleton, so they cannot be indifferent. This means that a unique pixep is played in every SPE.
\end{itemize}

Lemma \ref{lem:spe}  does not hold with non-strict preferences. For example, suppose there are two items $x,y$, Alice is indifferent between them while Bob strictly prefers $y$. Then the picking-sequence $A B$ has two SPEs --- one in which Alice picks $x$ and one in which she picks $y$ --- and they lead to two different allocations that are substantially different for Bob.

\end{document}